\documentclass{article}

\usepackage{amsthm,graphicx,amssymb,amsmath,paralist}
\usepackage[mathcal]{euscript}
\usepackage[symbol*]{footmisc}
\usepackage{graphicx}
\graphicspath{{figures/}}
\usepackage{subcaption}
\usepackage{hyperref}
\usepackage{cleveref}
\usepackage{xcolor}
\usepackage{xspace}
\usepackage{authblk}

\newcommand{\C}{\ensuremath{\mathcal{C}}\xspace}

\newtheorem{theorem}{Theorem}
\newtheorem{lemma}[theorem]{Lemma}
\newtheorem{conjecture}{Conjecture}

\crefname{theorem}{Theorem}{Theorems}
\crefname{section}{Section}{Sections}
\crefname{lemma}{Lemma}{Lemmas}
\crefname{figure}{Figure}{Figures}

\crefname{claim}{Claim}{Claims}
\newtheorem{corollary}{Corollary}
\crefname{corollary}{Corollary}{Corollaries}

\title{Folding Polyominoes with Holes into a Cube\thanks{A preliminary extended abstract appeared in the Proceedings of the 31st Canadian Conference on Computational Geometry \cite{aac+-fphc-19}.}}

\author[1]{Oswin Aichholzer}
\author[2]{Hugo A. Akitaya}
\author[3]{Kenneth C. Cheung}
\author[4]{Erik D. Demaine}
\author[4]{Martin L. Demaine}
\author[5]{S\'andor P. Fekete}
\author[5]{Linda Kleist}
\author[6]{Irina Kostitsyna}
\author[7]{Maarten L\"{o}ffler}
\author[8]{Zuzana Mas\'arov\'a}
\author[4]{Klara Mundilova}
\author[9]{Christiane Schmidt}

 \affil[1]{Institute for Software Technology, Graz University of Technology, {\tt oaich@ist.tugraz.at}}
 \affil[2]{Department of Computer Science, Tufts University, {\tt Hugo.Alves\_Akitaya@tufts.edu}}
  \affil[3]{NASA Ames Research Center, {\tt kenny@nasa.gov}}
  \affil[4]{CSAIL, Massachusetts Institute of Technology, {\tt$\{$edemaine,mdemaine,kmundi$ \}$@mit.edu}}
  \affil[5]{Department of Computer Science, TU Braunschweig, {\tt $\{$s.fekete,l.kleist$\}$@tu-bs.de}}
  \affil[6]{Mathematics and Computer Science Department, TU Eindhoven, {\tt i.kostitsyna@tue.nl }}
  \affil[7]{Department of Information and Computing Science, Universiteit Utrecht, {\tt m.loffler@uu.nl}}
  \affil[8]{IST Austria, Klosterneuburg, {\tt zuzana.masarova@ist.ac.at}}
  \affil[9]{Department of Science and Technology, Link\"{o}ping University, {\tt christiane.schmidt@liu.se}}

\date{}
  
\begin{document}

\maketitle

\begin{abstract}
When can a polyomino piece of paper be folded into a unit cube?
Prior work studied tree-like polyominoes,
but polyominoes with holes remain an intriguing open problem.
We present sufficient conditions for a polyomino with one or several holes to fold into a
cube, and conditions under which cube folding is impossible. In particular, we show that all but five special \emph{basic} holes guarantee foldability.
\end{abstract}

{\bf Keywords.}
Folding; Origami Folding;
Cube; Polyomino;
Polyomino with Holes; Non-Simple Polyomino



\section{Introduction}

\begin{figure}[htb]
\centering
 \includegraphics{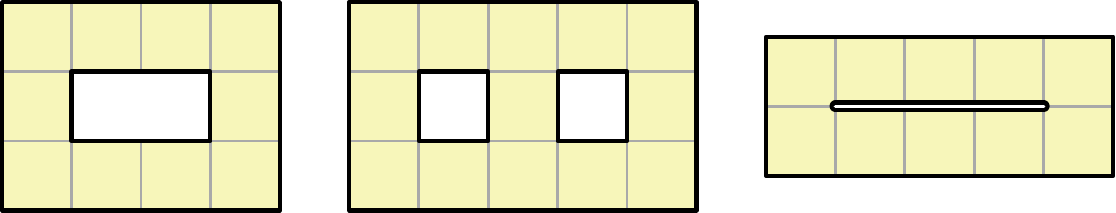}
\caption{Three polyominoes that fold along grid lines into a unit cube, from puzzles by Nikolai Beluhov \cite{Beluhov-2014}.}
 \label{puzzles}
\end{figure}

Given a piece of paper in the shape of a polyomino, i.e., a polygon in the plane formed by unit squares on the square lattice that are connected edge-to-edge, does it have a folded state
in the shape of a unit cube?  The standard rules of origami apply~\cite{GFALOP}; in particular, we allow each unit-square
face to be covered by multiple layers of paper.
Examples of this decision problem are given by the three puzzles by Nikolai Beluhov
\cite{Beluhov-2014} shown in \cref{puzzles}.
We encourage the reader to print out the puzzles and try folding them.

Prior work \cite{abdde-fpipc-18} studied this decision problem extensively,
introducing and analyzing several different models of folding.
Beluhov
\cite{Beluhov-2014} implicitly defined a \emph{grid model} with the puzzles
in \cref{puzzles}:
Fold only along grid lines of the polyomino;
allow only orthogonal fold angles%
\footnote{The fold angle of a fold measures the deviation from the flat (unfolded) state, i.e., $180^\circ$ minus the dihedral angle between the two incident faces.}
($\pm 90^\circ$ and $\pm 180^\circ$); and
forbid folding material strictly interior to the cube.
In this model, the prior work \cite{abdde-fpipc-18} characterizes which
\emph{tree-shaped} polyominoes (whose unit squares are connected edge-to-edge
to form a tree dual graph)
lying within a $3 \times n$ strip can fold into a unit cube,
and exhaustively characterizes which tree-shaped polyominoes of $\leq 14$
unit squares fold into a unit cube.

Notably, however, the polyominoes in \cref{puzzles} are not tree-shaped,
and their interior is not even simply connected:
The first puzzle has a hole, the second puzzle has two holes, and the
third puzzle has a degenerate (zero-area) hole or \emph{slit}.  Arguably, these holes are what
makes the puzzles fun and challenging.  Therefore, in this paper, we embark on
characterizing which polyominoes with hole(s) fold into a unit cube in the
grid model.
Although we do not obtain a complete characterization, we give many
interesting conditions under which a polyomino does or does not fold into a
unit cube.

The problem is sensitive to the choice of model.
The other main model that has been studied in past work
is the more flexible \emph{half-grid model},
which allows orthogonal and diagonal folds between half-integral points,
as in Figure~\ref{fig:45-diag}.
The prior work \cite{abdde-fpipc-18} shows that \textit{all} polyominoes of at
least ten unit squares can fold into a unit cube in the half-grid model,
leaving only a constant number of cases to explore,
which were tackled recently \cite{CubeFolding_CCCG2020}.
Therefore, we focus on the grid model, which matches the puzzles
of Beluhov \cite{Beluhov-2014}.

\begin{figure}[ht]
	\centering
	\includegraphics[scale=1.25]{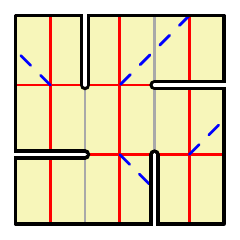}
	\caption{One way to fold a cube in the half-grid model, adapted from~\cite[Fig.~5(b)]{abdde-fpipc-18}. In all our figures, solid/red lines denote mountain folds, dashed/blue lines denote valley folds, light/grey lines denote grid lines, and bold/black lines denote the polyomino boundary.}
	\label{fig:45-diag}
\end{figure}

If we generalize the target shape from a unit cube to polycube(s),
there are polyominoes that fold in the grid model into all polycubes
of at most a given surface area \cite{StripsGrid_CGTA}.
If we further forbid overlapping unit squares (polyhedron unfolding/nets
instead of origami), this fold-all-polycubes problem has been studied for
small polycubes \cite{AloBosCol-CCGA-10}, and there is extensive work on
finding polyominoes that fold into multiple (two or three or more)
different boxes
\cite{AbeDemDem-CCCG-11,Mitani-Uehara-2008,Shirakawa-Uehara-2013,Uehara-box-survey,Xu2015}.

\subsection*{Our Results}

\begin{compactenum}
\item 
We show that any hole that is not one of five \emph{basic} shapes of holes (see Figure~\ref{fig:simple-holes1}) always guarantee that a polyomino containing the hole folds into a cube; see \cref{thm:singleHole} in \cref{subsec:fold-1hole}. Polyominoes with exactly one of the five basic holes only sometimes allow folding into a cube.
\item We identify combinations of two (of the five basic) holes that allow the polyomino to fold into a cube; see \cref{subsec:fold-2holes}.
\item We show that certain of the five basic holes or their combinations do not allow folding into a cube, that is, we show that subclasses of polyominoes with only specific basic hole(s) cannot be folded into a unit cube; see \cref{sec:not-fold}.
\item We present an algorithm that checks a necessary local condition for folding into a cube; see \cref{subsec:alg}.
\item Whether this condition also constitutes a sufficient condition remains an open question; see \cref{sec:concl}.
\item We conjecture that a slit of size 1 (see Figure~\ref{fig:simple-holes1}, second from left) never affects whether a polyomino can fold into a cube; see \cref{subsec:not-fold-slit}. However, we show that a slit of size 1 can be the deciding factor for foldability for larger polycubes.
\end{compactenum}

\section{Notation}\label{sec:not}
A \emph{polyomino} is a connected polygon $P$ in the plane formed by joining
together $|P|=n$ unit squares on the square lattice.
We refer to the vertices of the $n$ unit squares forming $P$
as the \emph{grid points} of $P$.
We view $P$ as an \emph{open} region (excluding its boundary) which includes
the $n$ open unit squares of the form $(x_i,x_i+1) \times (y_i,y_i+1)$ as well
as \textit{some} of the shared unit-length edges (and grid points)
among these $n$ unit squares.
Notably, we do not require $P$ to include the common edge between every
adjacent pair of squares; if such an edge is missing from $P$,
we call the edge a \emph{slit edge}.
But there must be at least $n-1$ unit-length edges in $P$ so that $P$ is
(interior-)connected.

A \emph{hole} of a polyomino $P$ is a bounded connected component of $P$'s
exterior, whose boundary is one of the connected components of $P$'s boundary
other than the outermost one.
We assume that $P$ has no holes that are just a single grid point,
because such holes do not affect foldability,
so we can fill them in (add them to~$P$).
We call a hole a \emph{slit} if it has zero area (and is more
than a single point), and thus consists entirely of one or more slit edges.
We call a hole \emph{basic} if it has one of the following shapes
(refer to see \cref{fig:simple-holes1}):
\begin{enumerate}
\item A unit square
\item A slit of size~1 (a single slit edge)
\item A slit of size~2 (L-shaped or straight)
\item A U-slit of size~3
\end{enumerate}

A \emph{unit cube \C} is a three-dimensional polyhedron with six unit-square faces and volume of $1$. We refer to the vertices of $\C$ as \emph{corners}.

In this paper, we study the problem of folding a given polyomino $P$ with holes to form $\C$, allowing 
creases along edges of the lattice with fold angles of
$\pm 90^\circ$ or $\pm 180^\circ$.
In all our figures, solid/red lines denote mountain folds, dashed/blue lines denote valley folds, light/grey lines denote grid lines, bold/black lines denote the polyomino boundary, bold dotted/purple lines denote creases folded by $\pm 180^\circ$, and thin dotted/purple lines denoted creases folded by $\pm 90^\circ$ (the purple/dotted creases could be mountain or valley).

Some crease patterns give numbers on the unit squares to indicate which face
they fold onto in a ``real-world'' six-sided die, where opposite faces sum up to~$7$.

\begin{figure}[ht]
	\centering
	\includegraphics[page=3]{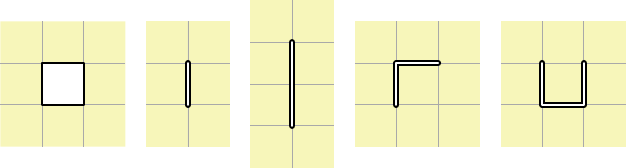}
	\caption{The five basic holes: a unit square, a slit of size~1, a straight slit of size~2, an L-slit of size~2, and a U-slit of size~3.}
	\label{fig:simple-holes1}
\end{figure}

\section{Polyominoes That Do Fold}\label{sec:fold}
In this section, we present polyominoes that fold. We start with polyominoes that contain a hole guaranteeing foldability.

\subsection{Single-Hole Polyominoes}\label{subsec:fold-1hole}

In this section, we show that any hole different from a basic hole guarantees foldability.

\newcommand{\simple}{basic}
\begin{theorem}\label{thm:singleHole}
If a polyomino $P$ contains a hole $h$ that is not \simple, then $P$ folds into a cube.
\end{theorem}
\begin{proof}
By the definition of \simple~hole, because $h$ is non-\simple, it must be a superset of either two unit squares, a unit square and a unit slit, or a slit of size 3 that is not a U-slit.
In the case it contains two unit squares sharing a grid point, $h$ must be a superset of one of the holes in Figure~\ref{fig:cases}~(a)--(b) up to rotation. 
If it contains a unit square and a unit slit sharing a grid point, then $h$ is a superset of Figure~\ref{fig:cases}~(e) up to reflection and rotation.
Else, $h$ must be a superset of the slits in Figure~\ref{fig:cases}~(c), (d), (f), (g) because those are all possible slits of size 3 that are not U-shaped up to reflection and rotation. 
Then, we distinguish the cases where $h$ contains 
\begin{compactitem}
\item  Two unit squares sharing an edge
\item  Two unit squares sharing a grid point
\item  A unit square and an incident slit
\item  A slit of length at least 3 (straight, zigzagged, L-shaped, or T-shaped)
\end{compactitem}

In a first step, we show that if $h$ contains one of the four above holes, we may assume that it contains exactly this hole. Let $h$ be a hole containing a hole $h'$ of the above type. By definition of a hole, $h$ needs to be enclosed by solid squares. Thus we can sequentially fold the squares of $P$ in columns to the left and right of~$h'$ on top of the columns directly left and right of~$h'$, respectively, as illustrated in \cref{fig:reduction}.
\begin{figure}[t]
	\centering
	\includegraphics[page=2]{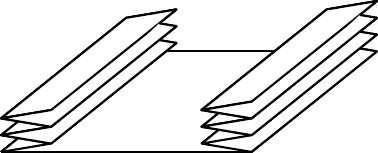}
	\caption{Folding strategy to reduce to seven cases.}
	\label{fig:reduction}
\end{figure}
Afterwards, we fold the squares of~$P$ in rows to the top and bottom of~$h'$ on top of the rows directly top and bottom of~$h'$, respectively.  We call the resulting polyomino~$P'$. Note that because~$h$ is a hole of $P$, all neighboring squares of~$h'$ exist in~$P'$. Thus we may assume that we are given one of the seven polyominoes depicted in \cref{fig:cases}, where striped  squares may or may not be present. Note that we can assume that no additional slits are present.

\begin{figure}[htbp]
	\centering
	\includegraphics[page=1]{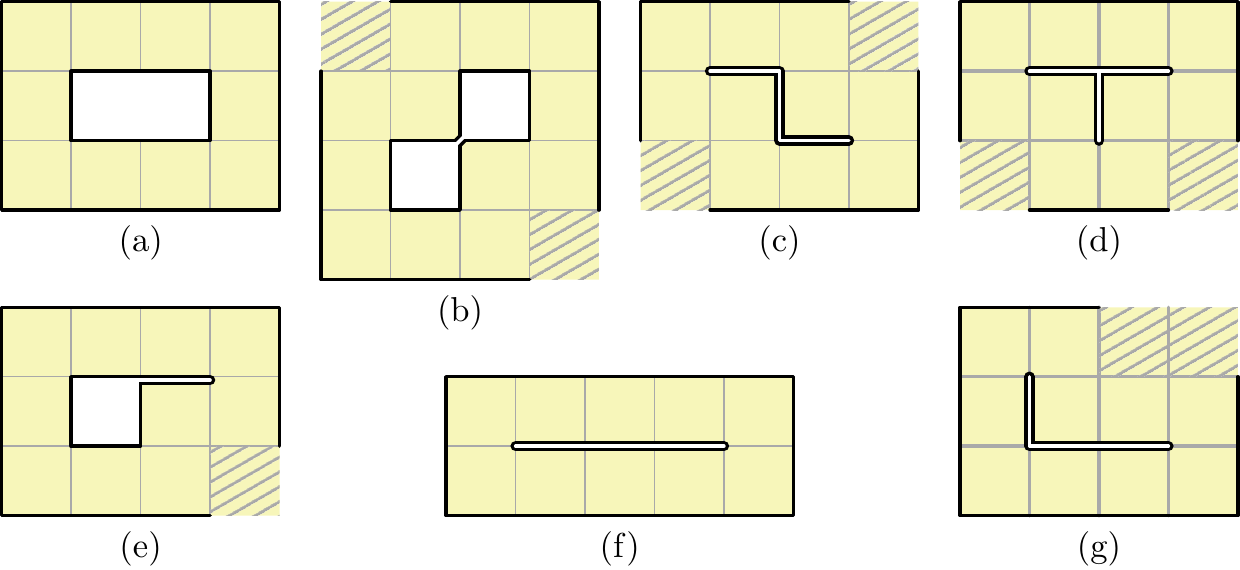}
	\caption{Any polyomino with a hole that is not \simple~can be reduced to one of the seven illustrated cases; striped squares may or may not be present.}
	\label{fig:cases}
\end{figure}

Secondly, we present strategies to fold the polyomino into $\C$.
Note that the case if $h'$ consists of two squares sharing only a grid point, we can fold the top row on its neighboring row and obtain the case where $h'$ consist of a square and an incident slit. For an illustration of the folding strategies for the remaining six cases consider \cref{fig:foldings}.
\begin{figure}[htbp]
	\centering
	\includegraphics[page=4]{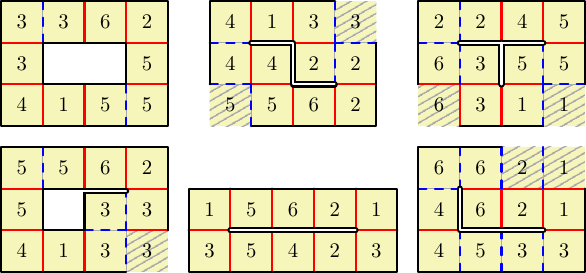}
	\caption{Crease pattern of cube foldings. Mountain folds are shown in solid red, valley folds in dashed blue. Squares with the same number cover the same face of the cube.}
	\label{fig:foldings}
\end{figure}
\end{proof}

\begin{figure}[hp]
	\centering
	\includegraphics[page=1]{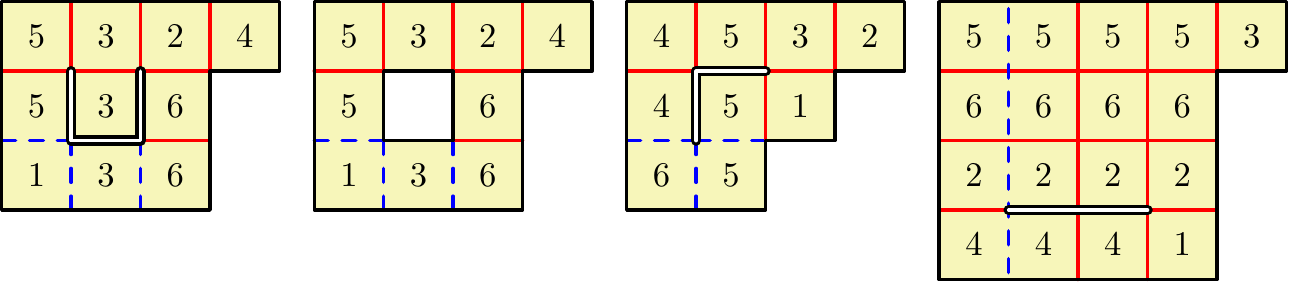}
	\caption{Four \simple~holes may be helpful. Mountain folds are shown in solid red, valley folds in dashed blue.}
	\label{fig:helpful}
\end{figure}

\subsection*{Are \simple~holes ever helpful?}
In fact, four of the five \simple~holes sometimes allow foldability, as illustrated in \cref{fig:helpful}.
Note that the U-slit of size 3 reduces to the square hole by a $\pm 180^\circ$ mountain or valley fold.

 In \cref{thm:slit1}, we show that the slit of size 1 never helps to fold a rectangular polyomino.
 Moreover, we show in \cref{lem:slitNEW1} that the crease pattern around the slit behaves as if the slit was nonexistent, i.e., the only option to make use of the slit is to push part of the polyomino through the slit.
 In fact, we conjecture that the slit of size~1 never helps to fold a polyomino into~\C.

\subsection{Combinations of Two Basic Holes}\label{subsec:fold-2holes}

In this section, we consider combinations of two \simple~holes that fold. For a polyomino with two holes, for which the lowest grid point $v$ of the upper hole has a higher y-coordinate than the upper grid point $w$ of the lower hole, we denote the number of unit-square rows between $w$ and $v$ as the \emph{number of rows between} the two holes. Analogously, we define the \emph{number of columns between} two holes.

\begin{theorem}
A polyomino with two vertical straight size-2 slits with at least two columns and an odd number of rows between them folds.
\end{theorem}
\begin{proof}
As in the previous section, we first fold all rows between the slits together to one row; this is possible because there is an odd number of rows between the slits. Then, all the surrounding rows and columns are folded towards the slits. Finally, we fold the columns between the slits to reduce their number to two or three. Depending on whether the number of columns between the slits was even or odd, this yields a shape as shown in \cref{fig:twoholes1}~(a) and~(b), respectively, where the striped squares may be (partially) present.
\begin{figure}[b!]
	\centering
	\includegraphics{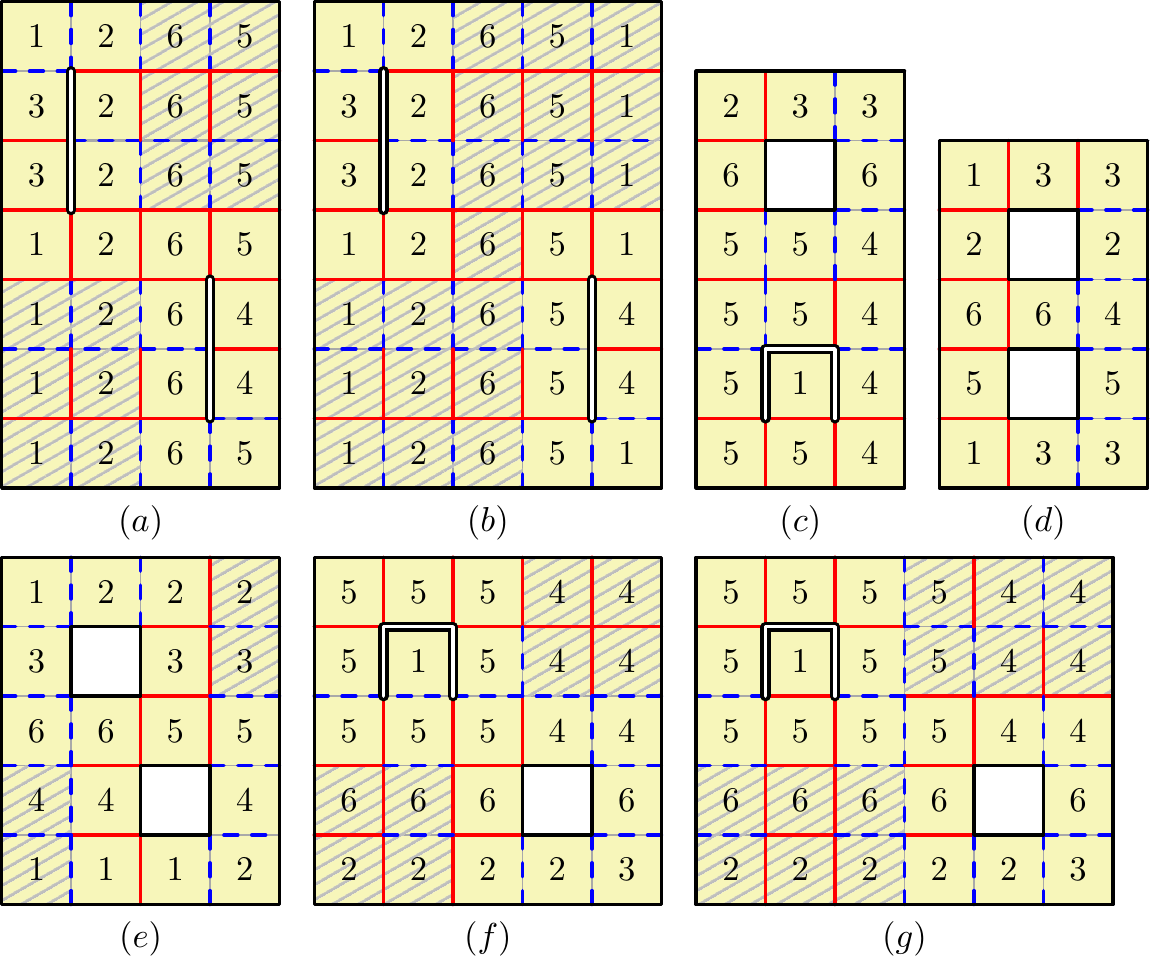}
	\caption{Combinations of two \simple~holes that are foldable with and without (part of) the striped region. Mountain folds are shown in solid red, valley folds in dashed blue. 
	}
	\label{fig:twoholes1}
\end{figure}
In all cases, the two shapes fold as indicated by the illustrated crease pattern. Note that in \cref{fig:twoholes1}~(b) the polyomino is of course connected, which implies that at least one square of the central column is part of the polyomino, i.e., a square with label 6 is used.
\end{proof}

If we have a polyomino with exactly two slits that have only one or no column between them, then the shape cannot be folded as can be verified by the algorithm of \cref{subsec:alg}.
In the following theorems we call a U-slit which has the open part at the bottom an A-slit. If the orientation of the U-slit is not relevant, then we call it a C-slit.

\begin{theorem}
A unit cube can be folded from any polyomino
with an A-slit and a unit-square hole/C-slit in the same column above it, having an even number of rows between them. 
\end{theorem}

\begin{proof}
We can assume that the upper hole is a unit square, as the flaps generated by a C-slit can always be folded away. 
Similar to before we fold away all surrounding rows and columns and reduce the number of rows between the A-slit and the unit-square hole to two. This yields the shape of \cref{fig:twoholes1}~(c), which can be folded as indicated by the crease pattern.
\end{proof}

Note that if the bottom slit is a U-slit, then the shape of \cref{fig:twoholes1}~(c) does not fold, again verified by the algorithm of \cref{subsec:alg}.

\begin{theorem}
A polyomino with an A-slit and a unit-square hole/C-slit below it and separated by an odd number of rows, folds, regardless in which columns they are.
\end{theorem}

\begin{proof}
As before, we assume that the lower hole is a unit square, fold away the surrounding rows and columns, and reduce the number of rows between the two slits/holes to one. Furthermore, we fold the columns between the slits/holes such that at most two columns remain between the two slits/holes. Consequently, we obtain one of the shapes shown in \cref{fig:twoholes1}~(d) to~(g). All of them fold, with or without the striped region. Note that the upper unit-square holes in \cref{fig:twoholes1}~(d) and~(e) can be replaced by an A-slit which can be folded away.
\end{proof}

Note that if the two holes are in the same or neighbored column(s) (\cref{fig:twoholes1}~(d) and~(e)), then independently of the orientation of the U-slits or whether they are unit-square holes, any combination folds, yielding Theorem~\ref{th:2us-holes}. In the other cases, the unit square incident to all three slit edges constitutes the only unit square that covers the face '1' in the unit cube.

\begin{theorem}\label{th:2us-holes}
 A polyomino with two unit-square holes which are in the same or in neighboured column(s) and have an odd number of rows between them folds.
\end{theorem}

\section{Polyominoes That Do Not Fold}\label{sec:not-fold}
In this section, we identify \simple~holes and combinations of \simple~holes that do not allow the polyomino to fold.
First, we present some results that show how the paper is constrained around an interior grid point $v$. In particular, we consider situations when the induced polyomino of the four unit squares $A, B, C, D$ incident to $v$ is connected; for an example consider \cref{fig:180creases}.

\begin{lemma}\label{le:43}
Four unit squares incident to a polyomino grid point $v$ for which the induced polyomino is connected,  cannot cover more than three faces of $\C$. 
\end{lemma}

\begin{proof}
 The grid point $v$ is incident to four unit squares in $P$. As grid points of $P$ are mapped to corners of $\C$ and all corners of $\C$ are incident to 3 faces, $v$ is incident to only 3 faces in $\C$. 
\end{proof}

\begin{lemma}\label{le:4by4}
Four unit squares incident to a grid point $v$ not on the boundary of $P$ cannot cover more than two faces of $\C$. In particular, at least two collinear incident creases are folded by $\pm 180^\circ$.
\end{lemma}
\begin{proof}
Let $A$, $B$, $C$, and $D$ be the unit squares incident to $v$ in circular order; see the left of \cref{fig:180creases}.
By \cref{le:43}, $A$, $B$, $C$, and $D$ cover at most three faces of $\C$. Hence, at least two unit squares map to the same face of $\C$; these can be edge-adjacent or diagonal. 

In the first case, assume without loss of generality that $A$ and $B$ map to the same face. Hence, the crease between them must be folded by $\pm 180^\circ$. Then $C$ and $D$ must also map to the same face of $\C$ to maintain the paper connected. Consequently, the crease between $C$ and $D$ is  folded by $\pm 180^\circ$.

In the latter case, let without loss of generality $A$ and $C$ map to the same face of $\C$. As they are both incident to $v$, only two options of folding those two unit squares on top of each other exist. Either the edge between $A$ and $B$ gets folded on top of the edge between $B$ and $C$, this leaves a diagonal fold on $B$, a contradiction, or the edge between $A$ and $D$ gets folded on top of the edge between $B$ and $C$, which results in $D$ being mapped to $C$, and those are two adjacent unit squares, by the above argument two collinear incident creases must be folded by $\pm 180^\circ$.
\end{proof}

\begin{figure}[htb]
	\centering
	\includegraphics[page=1]{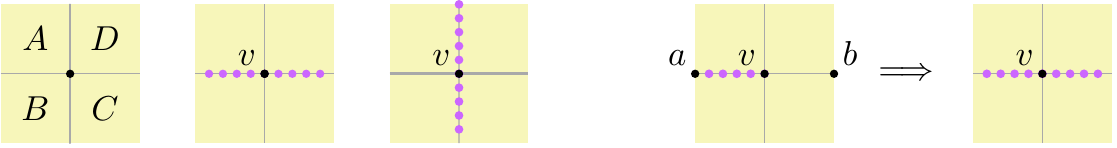}
	\caption{Illustration of \cref{le:4by4,le:4by4_2}. The thick dotted/purple lines represent creases fold by $\pm 180^{\circ}$; they could be mountain or valley.}
	\label{fig:180creases}
\end{figure}

\begin{lemma}\label{le:4by4_2}
Consider a grid point $v$ that is not on the boundary of a polyomino $P$ that folds into \C. If one crease of $v$ is folded by $\pm 180^{\circ}$, then the incident collinear crease is also folded by $\pm 180^{\circ}$.
\end{lemma}
\begin{proof}

Without loss of generality, we show that if  the left horizontal crease of $v$ is folded by $\pm 180^{\circ}$, the same holds for the right horizontal crease. We denote the left and right adjacent grid points of $v$ by $a$ and $b$, respectively, as indicated in \cref{fig:180creases}, right.

Suppose for a contradiction, that the right crease is not folded by $\pm 180^{\circ}$. Then, by  \cref{le:4by4}, both vertical creases are folded by $\pm 180^{\circ}$. In particular, $a$ and $b$ are mapped to the same corner of \C and thus the edges $av$ and $bv$ coincide. Hence, because $av$ is folded by $\pm 180^{\circ}$, $bv$ is also folded by $\pm 180^{\circ}$.
\end{proof}

\cref{le:4by4,le:4by4_2} imply that:
\begin{corollary}\label{cor:rows}
Let $k,n\geq2$ and let $P$ be a polyomino containing a rectangular $k\times n$-subpolyomino $P'$ whose interior does not contain any boundary of $P$. Then, in every folding of $P$ into \C, all collinear creases of $P'$ are either folded by  $\pm 90^{\circ}$ or by $\pm 180^{\circ}$. Moreover, either all horizontal or all vertical creases of $P'$ are folded by  $\pm 180^{\circ}$; see \cref{fig:reducing}.
\end{corollary}
\begin{proof}
First, suppose for a contradiction that there exist two collinear creases, one of which is folded by $\pm 90^{\circ}$ and the other by $\pm 180^{\circ}$. Then there also exists an interior grid point of $P'$ where the crease type of the two collinear edges changes from $\pm 90^{\circ}$ to $\pm 180^{\circ}$. However, by \cref{le:4by4_2}, if one is folded by $\pm 180^{\circ}$, then both are. A contradiction.

Second, suppose that not all horizontal creases are folded by $\pm 180^{\circ}$. Then, by the first statement, there exists a row in which no grid point has a horizontal edge that is folded by $\pm 180^{\circ}$. By \cref{le:4by4}, all vertical creases incident to the grid points of this row are folded by $\pm 180^{\circ}$. Because all collinear edges behave alike, it follows that all vertical creases are folded by $\pm 180^{\circ}$.
\end{proof}

\begin{figure}[htb]
	\centering
	\includegraphics[page=2]{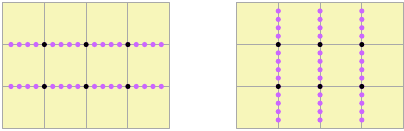}
	\caption{Illustration of \cref{cor:rows}.}
	\label{fig:reducing}
\end{figure}

\begin{corollary}\label{cor:rect}
Let $P$ be a rectangular $k\times n$-polyomino without any holes, then $P$ does not fold into $\C$.
\end{corollary}

\subsection{Polyominoes with Unit Square, L-Shaped, and U-Shaped Holes}\label{subsec:pol-ush-L-U}

We begin by examining the possible foldings of a polyomino containing a unit-square hole.
Suppose that a given polyomino $P$ with a unit-square hole~$h$ folds into a cube.
Furthermore, let the shape of $h$ no longer be a square in the folded state; we say hole $h$ is folded in a \emph{non-trivial} way. For an example consider Figures~\ref{fig:non-trivial} and \ref{fig:one-edge-fold}. Then, in the folded state, either all edges of $h$ are mapped to the same edge of \C, or two pairs of edges are glued forming an L-shape.
In the following, we show that if $P$ folds into \C, the first case is impossible, while the second produces a  specific crease pattern around~$h$.

\begin{figure}[h]
\centering
\includegraphics{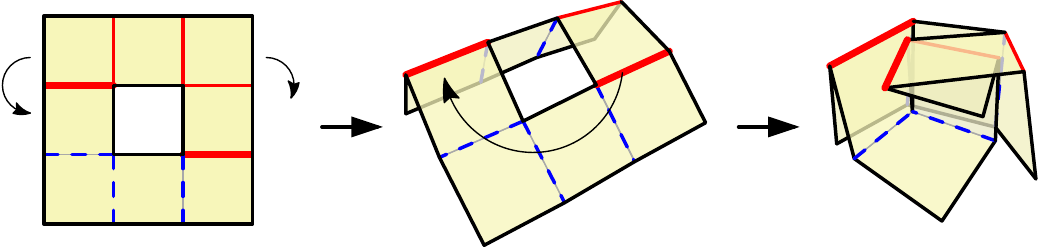}
\caption{An example of a non-trivial fold of a $3\times 3$ square with a unit square hole. The crease pattern is a special case of the one in \cref{fig:hole-crease-pattern}.}
\label{fig:non-trivial}
\end{figure}

\begin{lemma}\label{lem:SquareHole1}
The four edges of a unit-square hole $h$ of a polyomino $P$ that folds into~$\C$ are not mapped to the same edge of \C in the folded state.
\end{lemma}
\begin{proof}
We denote the four unit squares of the polyomino edge-adjacent to $h$ by $A$, $B$, $C$, and $D$, and the four unit squares adjacent to $h$ via a grid point as $F_{1}$, $F_{2}$, $F_{3}$, and $F_{4}$, as illustrated in \cref{fig:one-edge-fold}. Assume for a contradiction that all edges of~$h$ are mapped to the same edge of $\C$. Consider $A$, $F_{1}$, and $B$ in the folded state. As the two corresponding edges of $h$ are glued together, the three faces must be pairwise perpendicular. The similar statement holds for the triples $\{B,F_{2},C\}$, $\{C,F_{3},D\}$, and $\{D,F_{4},A\}$. This results in a configuration as illustrated in the right of \cref{fig:one-edge-fold}.

\begin{figure}[htb]
	\centering
	\includegraphics[page=1]{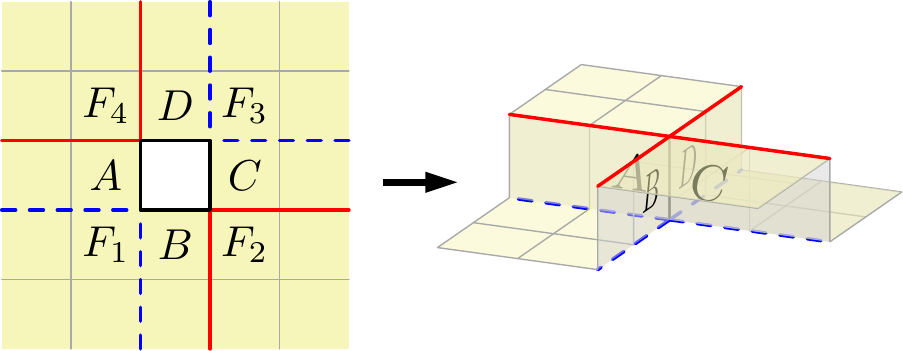}
	\caption{Four edges of a square hole glued together.}
	\label{fig:one-edge-fold}
\end{figure}

Because the faces $A,B,C$ share an edge of \C in the folded state such that $A$ and~$B$, as well as $B$ and $C$ are perpendicular, $A$ and $C$ must cover the same face of \C. Likewise, $B$ and $D$ cover the same face of \C. If $P$ folds into \C, then $F_{1}$ and $F_{3}$, as well as $F_{2}$ and $F_{4}$ are mapped to the same faces of \C.
Suppose, without loss of generality, that in the folded state $A$ lies in a more outer layer than $C$. Then, $F_{1}$ and $F_{4}$ are in a more outer layer than $F_{3}$ and $F_{2}$, respectively. Thus, $B$ connects the more inner layer of $F_{2}$ to the more outer layer of $F_{1}$, and at the same time $D$ connects the inner layer of $F_{3}$ to the outer layer of $F_{4}$. Hence, $B$ and $D$ intersect, which is impossible. Therefore, if the polyomino folds into a cube, the four edges of a square hole cannot all be mapped to the same edge of \C.
\end{proof}

It follows that the only non-trivial way to glue the edges of a square hole $h$ of a polyomino folded into a cube is to form an L-shape. We use this to show the following fact:

\begin{lemma}\label{lem:SquareHoleCreasePattern}
Let $P$ be a polyomino with a unit-square hole that folds into~\C.
 In every folding of $P$ into~\C where $h$ is folded non-trivially (i.e., $h$ is not mapped to a square), the crease pattern of the unit squares incident to $h$ is as illustrated in the right image of Figure~\ref{fig:hole-crease-pattern} (up to rotation and reflection).
\end{lemma}
\begin{proof}
Suppose the four edges of $h$ are not mapped to distinct edges of~\C. Then, by \cref{lem:SquareHole1}, the four edges are not mapped to the same edge, but to two edges forming an L-shape. This effectively amounts to gluing a pair of diagonal grid points of the hole.

Let $a$, $b$, $c$, and $d$ be the four grid points of $h$, and suppose $a$ and $c$ are mapped to the same corner of \C when folding $P$ into \C; see also the left image of \cref{fig:hole-crease-pattern}.
\begin{figure}[t]
	\centering
	\includegraphics[page=1]{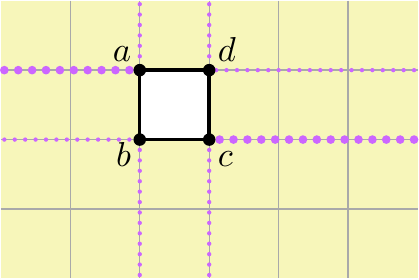}
	\hfil
	\includegraphics[page=2]{hole-crease-pattern-color.pdf}
	\caption{Thinner dotted/purple lines indicate creases folded by $\pm 90^{\circ}$, while thicker dotted/purple lines indicate creases folded by $\pm 180^{\circ}$. Left: crease pattern around a unit-square hole folding into an L-shape when grid points $a$ and $c$ are mapped to the same corner of \C; creases shown in purple can be both mountain or valley. Right: numbers indicate the face of the cube in the folded state; mountain folds are shown as red solid lines, and valley folds as blue dashed lines.}
	\label{fig:hole-crease-pattern}
\end{figure}

Consider the crease pattern around $h$.
We shall only focus on the angles of the creases and not the type of the fold, as there may be (and will be) other creases in $P$ affecting the type of the creases under our consideration.
Observe that the three faces incident to each of the grid points $b$ and $d$ are pairwise perpendicular, they form a corner of a cube.
Thus, the creases emanating from $b$ and $d$ are all folded by $\pm 90^{\circ}$.
Further observe that the three unit squares around each of the grid points $a$ and $c$ fold into two faces of a cube, thus, leading to one of the creases being folded by $\pm 90^{\circ}$ and the other folded by $\pm 180^{\circ}$.
Finally, the two creases folded by $\pm 180^{\circ}$ are parallel to each other.
Indeed, consider the right side of \cref{fig:one-edge-fold}.
For a crease to form an L-shape one of the two dashed blue lines must fold by $\pm 180^{\circ}$, which corresponds to two parallel creases in the unfolded state.
Therefore, the crease pattern in \cref{fig:hole-crease-pattern} (left) is the only pattern of creases (up to rotation and reflection) around a non-trivially folded square hole.
\cref{fig:hole-crease-pattern} (right) shows the faces of the corresponding crease pattern, and \cref{fig:non-trivial} shows the folding process of this crease pattern.
\end{proof}

With the help of \cref{lem:SquareHoleCreasePattern}, we can show that several types of polyominoes with unit-square holes do not fold into \C.

\begin{theorem}\label{thm:RectangleSquareHole}
If $P$ is a rectangle with exactly one unit-square hole~$h$, then $P$ does not fold into $\C$.
\end{theorem}
\begin{proof}
First note that $h$ is folded non-trivially, otherwise $P$ corresponds to a rectangle which does not fold into \C (Corollary~\ref{cor:rect}). Therefore, by \cref{lem:SquareHoleCreasePattern}, the crease pattern around $h$ is as depicted in \cref{fig:hole-crease-pattern}. Note that, on each side of~$h$, there exists a fold by $\pm 90^\circ$.

Consider the rectangle $R$ obtained by cutting $P$ by the top edge of $h$ and deleting the part below. If $R$ has a height of at least 2, then by \cref{cor:rows}, either all vertical or all horizontal creases are folded by $\pm 180^\circ$. 
In the first case, in particular the creases incident to $h$ are folded by  $\pm 180^\circ$. However, this is a contradiction to the crease pattern around $h$ in which each side of $h$ has fold by $\pm 90^\circ$.
Consequently, all horizontal edges are folded by $\pm 180^\circ$. This corresponds to folding $R$ on top of the row above $h$. In particular, $P$ is foldable into \C if and only if the polyomino $P'$ obtained from $P$ by cutting-off all rows above $h$ is foldable. Hence, we consider $P'$. 

Likewise, we treat all other sides of $P'$ and obtain the polyomino $P''$ consisting of a $3\times 3$-rectangle with a central unit-square hole; see also \cref{fig:hole-crease-pattern} (right). In particular,
$P$ is foldable (if and) only if $P''$ is foldable into \C.

Because $h$ is folded non-trivially, the crease pattern of $P''$ is given by \cref{fig:hole-crease-pattern}. Note that in the folded state $P''$  covers only 5 faces and, hence,  $P''$ does not fold into~\C.
\end{proof}

A similar result holds for rectangular polyominoes with two unit-square holes.

\begin{theorem}\label{thm:twoholeseven}
A rectangle with exactly two unit-square holes in the same row does not fold into \C if the number of columns between the holes is even.
\end{theorem}

\begin{proof}
Note that if the polyomino can be folded into~\C, both holes must be folded non-trivially: If one hole behaves as a square in the folded state, i.e., is folded trivially, the polyomino is effectively reduced to a rectangle with one \simple~hole. However, by  \cref{thm:RectangleSquareHole}, this does not fold into~\C. Consequently, both holes are folded non-trivially.

Therefore, by \cref{lem:SquareHoleCreasePattern}, the crease pattern around the two holes is as depicted in \cref{fig:hole-crease-pattern}.
%
Consider the $3\times 2k$-rectangle $R$ between the two holes (with $k\ge 1$).
By the above observation, at least one horizontal edge of $R$ is folded by $\pm 90^{\circ}$. Consequently, \cref{cor:rows} implies that all vertical edges are folded by $\pm 180^{\circ}$. 
In particular, every unit square of $R$ is mapped to the same face of \C as the leftmost (or rightmost) unit square in the same row of $R$. This reduces the polyomino to one with~$R$ being a $3\times 2$-rectangle.
We will show that the squares of $P$ neighbouring the two holes are not able to cover~\C, 
that is, it remains to show that the polyomino $P$, depicted in \cref{fig:two-holes-even}, does not fold into~\C. 

\begin{figure}[tb]
	\centering
	\includegraphics[page=1]{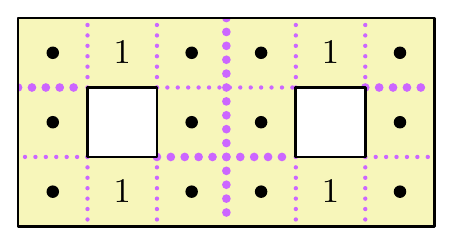}
	\hfil
	\includegraphics[page=2]{two-holes-even-color.pdf}
	\caption{A polyomino that does not fold into a cube.}
	\label{fig:two-holes-even}
\end{figure}

 Consider the left $3\times 3$ block of $P$. If the two parallel creases folded by $\pm 90^{\circ}$ are vertical, then the right $3\times 3$ block will also have the two parallel creases folded by $\pm 90^{\circ}$ run vertical; see \cref{fig:two-holes-even} (left). Then, the four unit squares above and below the two holes match to the same face on $\C$. Denote it as~`1'. Observe that the rest of the unit squares share a grid point with `1' and thus cannot cover the face on $\C$ opposite to `1'. 
 
 In the second case, when the two parallel creases folded by $\pm 90^{\circ}$ of the left block are horizontal, then they extend into the right $3\times 3$ block by \cref{cor:rows}. Refer to \cref{fig:two-holes-even} (right). Then, the four unit squares to the left and to the right of the two holes match to the same face on $\C$, which we denote by `1'. As before, every unit square of $P$ shares a grid point with `1' and thus the face opposite to `1' on $\C$ cannot be covered.
\end{proof}

\noindent \textbf{\emph{ Remark.}}
Note that the arguments of \cref{lem:SquareHoleCreasePattern} and \cref{thm:RectangleSquareHole,thm:twoholeseven} extend to an L-slit of size $2$, and a U-slit of size $3$. The resulting crease patterns are illustrated in \cref{fig:hole-crease-pattern2}.
\begin{figure}[htb]
	\centering
	\includegraphics[page=3]{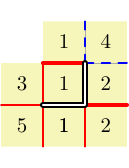}
	\hfil
	\includegraphics[page=4]{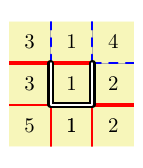}
	\caption{Crease pattern around an L-slit (left) and a U-slit (right).
	Numbers indicate the face of the six-sided die in the folded state; thinner lines denote creases folded by $\pm 90^{\circ}$; thicker lines denote creases folded by $\pm 180^{\circ}$; mountain folds are drawn solid/red; and valley folds are drawn dashed/blue.
	}
	\label{fig:hole-crease-pattern2}
\end{figure}
\medskip

These insights help to obtain the following result:

\begin{theorem}
\label{thm:two-holes-far}
Let $P$ be polyomino with two holes, which are both either a unit square, or an L-slit of size $2$, or a U-slit of size $3$, such that (1) $P$ contains all the other cells of the bounding box of the two holes and (2) the number of rows and the number of columns between the holes is at least $1$. In every folding of $P$ into~\C, the two holes are not both folded non-trivially. 
\end{theorem}
\begin{proof}
If $P$ contains a unit-square hole that is not folded non-trivially, then, by \cref{lem:SquareHoleCreasePattern}, the crease pattern in the neighborhood of the hole is as depicted in  \cref{fig:hole-crease-pattern}. Likewise, if $P$ contains an L-slit of size 2 or a U-slit of size 3 that is folded non-trivially, the crease pattern in the neighborhood of the hole is as depicted in  \cref{fig:hole-crease-pattern2}. Note that on each side of the crease patterns in the neighborhood of the holes, there exists a crease folded by $\pm 90^{\circ}$.

We turn the paper such that the left hole is above the right hole as in \cref{fig:two-hole-crease-pattern} and consider the rectangular region $R$ to the right of the left hole and above the right hole.

\begin{figure}[htb]
	\centering
	\includegraphics{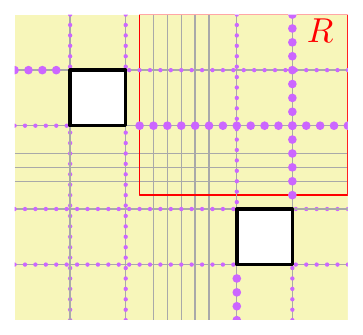}
	\caption{Two unit-square holes with at least one row and column in between, if folded non-trivially imply two perpendicular creases folded by $\pm 90^{\circ}$ (drawn as thin dotted lines).}
	\label{fig:two-hole-crease-pattern}
\end{figure}

Because each side of the crease patterns in the neighborhood of the holes has a crease folded by $\pm 90^{\circ}$ (Lemma~\ref{lem:SquareHoleCreasePattern}), $R$ contains a vertical and a horizontal crease folded by $\pm 90^{\circ}$. By \cref{cor:rows}, all collinear creases are also folded by $\pm 90^{\circ}$. Hence, there exists a grid point in $R$ for which all incident creases are folded by $\pm 90^{\circ}$, yielding a contradiction to \cref{le:4by4}.
%
\end{proof}

\subsection{Polyominoes with a Single Slit of Size 1}\label{subsec:not-fold-slit}

In the following, we show that a slit hole of size~1 does not help in folding a rectangular polyomino into $\C$.
We start with a lemma:

\begin{lemma}\label{lem:slitNEW1}
In every folding of a polyomino $P$ with a slit hole of size 1, the crease pattern behaves as if the slit hole was nonexistent.
\end{lemma}
\begin{proof}
To prove the lemma we examine the local neighborhood of the slit and analyze the possible folding patterns we can obtain between adjacent faces. More specifically, we consider the six unit squares $A$, $B$, $C$, $D$, $E$ and $F$ of $P$ that are incident to the slit hole of size 1 as illustrated in \cref{fig:slit1}.
We distinguish two cases: The crease between $A$ and $F$ is folded by $\pm 90^\circ$ or $\pm 180^\circ$.

\begin{figure}[ht]
	\centering
	\includegraphics[page=2]{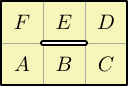}
	\caption{A polyomino with a slit hole of size one.}
	\label{fig:slit1}
\end{figure}

If the $AF$ crease is folded by $\pm 90^\circ$, we must further distinguish if the $EF$ crease is folded by $\pm 90^\circ$ or by $\pm 180^\circ$.
If the $EF$ crease is folded by $\pm 180^\circ$, then the slit edge is mapped to the edge between $AF$, fixing that $B$ maps to $A$. Hence, this corresponds to a crease folded by $\pm 90^\circ$ of the slit edge.

By symmetry, we may assume that both the $AB$ crease and the $EF$ crease are folded by $\pm 180^\circ$. This implies that $B$ and $E$ cover the same face in such a way that the top edge of $B$ is mapped to the left edge of $E$. However, then the bottom left corner of $D$ is also mapped to the top left corner of $E$. A contradiction. Consequently, this is impossible.

If the $AF$ crease is folded by $\pm 180^\circ$, then $A$ and $F$ cover the same face and, in particular, their left edges are mapped to the same edge such that the top edge of $F$ and the bottom edge of $A$ coincide. This implies that the left edge of $E$ and the left edge of $B$ also coincide such that the top edge of E and the bottom edge of $B$ coincide. This corresponds to crease folded by $\pm 180^\circ$ of the slit edge.

This shows that the slit edge is a crease folded by $\pm 90^\circ$ or by $\pm 180^\circ$. Hence, the crease pattern behaves as if the slit hole was nonexistent.
\end{proof}

\begin{theorem}\label{thm:slit1}
If $P$ is a rectangle with exactly one slit of size~1, then $P$ does not fold into \C.
\end{theorem}
\begin{proof}
By \cref{lem:slitNEW1}, the crease pattern behaves as if the slit was nonexistent, i.e., as if $P$ was a rectangle. By \cref{cor:rows}, all horizontal or vertical creases are folded by $\pm 180^\circ$, reducing $P$ to a rectangle of height or width 1, which does not fold into \C.
\end{proof}

Furthermore, we conjecture that the slit of size~1 never is the deciding factor for foldability.
\begin{conjecture}\label{conj:slit}
Let polyomino $P'$ be obtained from a polyomino $P$ by adding a slit~$s$ of size 1. If $P'$ folds into $\C$, then $P$ folds into $\C$ as well. 
\end{conjecture}

We note that this is not true for arbitrarily large \emph{polycubes} (connected three-dimensional polyhedron that are formed by a union of face-adjacent unit cubes on the cube lattice):
\begin{lemma}\label{le:slit-polycube}
There exists a polyomino $P$ with a slit $s$ of size 1 and a polycube $Q$, such that $P$ can be folded into $Q$, but the polyomino $P'$ without $s$ cannot be folded into $Q$. That is, for larger polycubes, a slit of size 1 can be the deciding factor for foldability.
\end{lemma}
\begin{proof}

\begin{figure}[htb]
\centering
\includegraphics[scale=1]{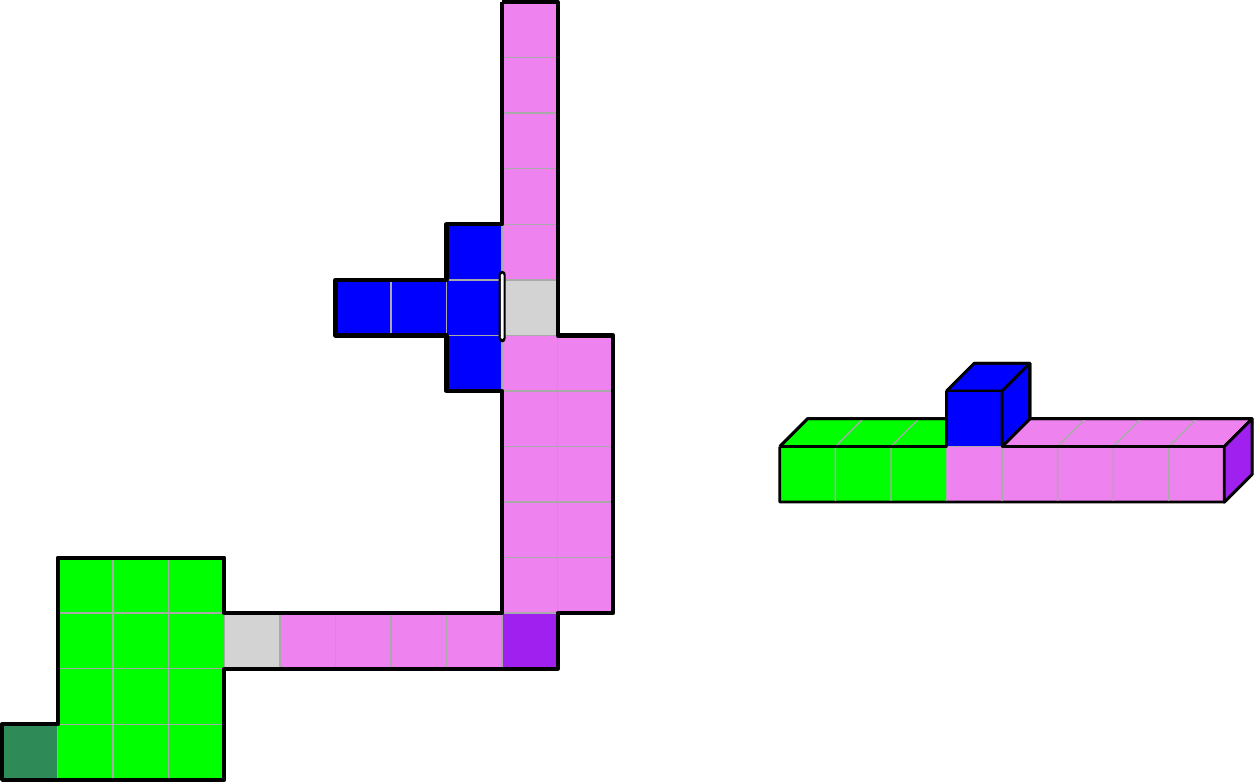}
\caption{Left: Polyomino $P$ (with slit $s$), right: polycube $Q$. Colors of $P$ coincide with the parts of the same color in $Q$, the light gray unit squares are not mapped to outer faces of $Q$.}
\label{fig:polycube-slit}
\end{figure}
We consider the polyomino $P$ and the polycube $Q$ from Figure~\ref{fig:polycube-slit}. $P$ has 40 unit squares, and $Q$ has 38 faces. Therefore, 38 out of the 40 unit squares will be the faces of $Q$ when folded, hence, at most two unit squares of $P$ may be folded on top of other unit squares to obtain $Q$.

$P$ contains two rectangular $k\times n$-subpolyominoes with $k,n\geq 2$ that do not contain $s$: a rectangular $3\times 4$-subpolyomino in the lower left (green) and a rectangular $5\times 2$-subpolyomino (pink). Similar to the proof for~\cref{cor:rows}, we know that there do not exist two collinear creases in these rectangular subpolyominoes, one of which is folded by $\pm 90^{\circ}$ and the other by $\pm 180^{\circ}$. Hence, if we fold a crease in those rectangular subpolyominoes by some angle, all other collinear creases (in the same row or column) are also folded by the same angle.
Observe that the surface of $Q$ does not contain any $2\times2$ flat squares. Hence, for every grid point contained in a rectangle at least the vertical or horizontal creases are folded by some angle.

Assume that we fold the vertical crease of length 5 in the pink rectangular $5\times 2$-subpolyomino by $\pm 180^{\circ}$. Then 5 unit squares would be folded on top of other unit squares in $Q$, again a contradiction. If, on the other hand, we fold a horizontal crease of length 2 in the rectangular $5\times 2$-subpolyomino by $\pm 180^{\circ}$, then all other unit squares need to appear as a face of $Q$. Similarly, if all of these creases would be $\pm 90^{\circ}$, again 2 unit squares would be folded on top of other unit squares. However, there are unit squares attached at the bottom and the top of the $5\times 2$-subpolyomino, which in that case cannot cover separate faces (5 unit squares from the pink rectangular $5\times 2$-subpolyomino plus these two adjacent unit squares can cover at most 4 faces of $Q$), which would yield further overlap, a contradiction.
Hence, the crease of length 5 must be $\pm 90^{\circ}$, thus, this will constitute part of the row of eight unit cubes in $Q$.

Analogously, assume that we fold any of the horizontal or vertical creases of the green rectangular $3\times 4$-subpolyomino by $\pm 180^{\circ}$. Hence, 3 or 4 of the unit squares would be folded on top of other unit squares in $Q$, a contradiction. Consequently, all existing folds must be $\pm 90^{\circ}$.

Assume that we fold all vertical creases in the green rectangular $3\times 4$ subpolyomino $\pm 90^{\circ}$. This would yield 3 faces of a tube-like $4\times 1 \times 1$-polycube for which the $1\times 1$ top and bottom faces and one of the $4\times 1$ faces are missing. However, together with (part of) the pink $5\times 1\times1$ polycube, this would yield a row of nine unit cubes, which cannot be combined for $Q$. Hence, all horizontal creases must be $\pm 90^{\circ}$.

Consequently, the green rectangular $3\times 4$-subpolyomino can be folded in a tube-like $3\times 1 \times 1$-polycube for which the $1\times 1$ top and bottom faces are missing. If we did not use the dark-green leftmost bottom unit square to cover one of these faces, this closing face would need to be a unit square of the remaining $27 (=40-(3\times 4+1))$ unit squares, however, then three unit squares must be folded on top of unit squares of the folded $3\times 4$-subpolyomino, a contradiction to the number of faces of $Q$ again. 

Hence the $3\times 1 \times 1$ polycube with one $1\times 1$-face missing (obtained from the green rectangular $3\times 4$-subpolyomino and the adjacent dark-green unit square), must cover the left $3\times 1 \times 1$-subpolycube of $Q$.

Then, the only part of $P$ that can be folded into the blue attached unit cube is the blue T-shape. 

The vertical unit-square row on top and below that T has length 5, hence, it must cover a part of the right $5\times 1\times 1$ of $Q$ (again, otherwise too many unit squares would be folded on top of each other). 

We obtain this only when using the slit of size 1 (we push the green $3\times 4$-subpolyomino and the adjacent dark-green unit square through the slit and unfold then again).
\end{proof}

\subsection{An Algorithm to Check a Necessary Local Condition for Foldability}\label{subsec:alg}
Consider the following local condition: let $s$ be a unit square in a
polyomino $P$ such that the mapping between grid points of $s$ and
corners of a face of $\C$ has been fixed. Then, for every unit
square $s'$ adjacent to $s$, there are two possibilities on how to map its four
grid points onto $\C$: the two grid points shared by $s$ and $s'$ must
be mapped consistently and for the other two grid points there are two
options depending on whether $s'$ is folded by $\pm 90^{\circ}$
to an adjacent face of $\C$, or whether it is folded by $\pm 180^{\circ}$
to the same face of $\C$. 

The algorithm below checks whether there exists a mapping between
all grid points of unit squares of $P$ to corners of $\C$ such that the
above condition holds for every pair of adjacent polyomino squares
of $P$.
\begin{compactenum}
\item Run a breadth-first search on the polyomino unit squares, starting with
the leftmost unit square in the top row of $P$ and continue via adjacent
unit squares. This produces a numbering of polyomino unit squares in which each
but the first unit square is adjacent to at least one unit square with smaller
number. See Figure~\ref{fig:algo-step1} for an example.
\item Map grid points of the first unit square to the bottom face of $\C$. Extend
the mapping one unit square at a time according to the numbering, respecting
the local condition (that is, in up to two ways). Track all such partial
mappings.
\end{compactenum}

\begin{figure}[htb]
\centering
\includegraphics[scale=1]{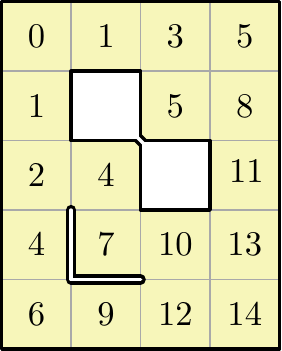}
\caption{Example of Step 1 of the algorithm. It shows the numbering of polyomino squares produced by the breadth-first search.}
\label{fig:algo-step1}
\end{figure}

The algorithm is exponential, because unless inconsistencies are produced, the number of possible partial mappings doubles with every polyomino unit
square. Nevertheless, it can be used to show non-foldability for small
polyominoes: if no consistent mapping exists for a polyomino, then
the polyomino cannot be folded onto $\C$. On the other hand, any consistent
grid-point mapping covering all faces of $\C$ obtained by the algorithm
that we tried could in practice be turned into a folding. However,
we have not been able to prove that this is always the case.

The algorithm above was used to prove that polyominoes in \cref{fig:simple_slits_cases}
do not fold, as well as it aided us to find the foldings of polyominoes in  \cref{fig:twoholes1}. An implementation of the algorithm is available at the following site \url{http://github.com/zuzana-masarova/cube-folding}.

\begin{figure}[htb]
\centering
\includegraphics{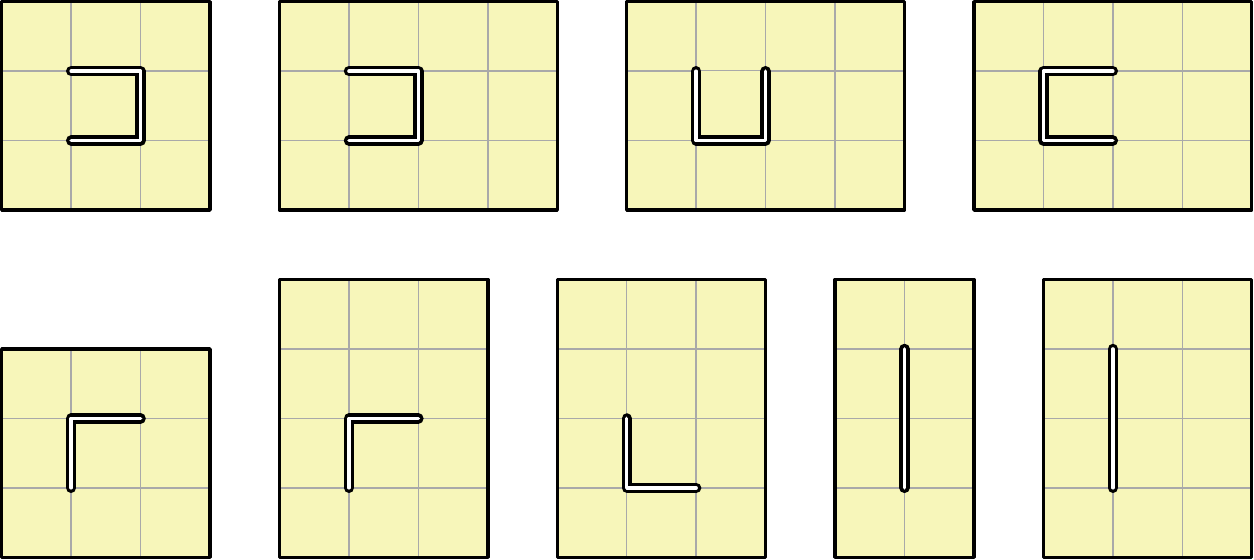}
\caption{These polyominoes with single L-, U- and straight size-2 slits do not
fold into a cube. }
\label{fig:simple_slits_cases}
\end{figure}

\section{Conclusion and Open Problems}\label{sec:concl}

We showed that, if a polyomino $P$ does contain a non-\simple~hole, then $P$ folds into $\C$. Moreover, we showed that a unit-square hole, size-2 slits (straight or L), and a size-3 U-slit sometimes allow for foldability.

Based on the presented results, we created a font of 26 polyominoes with slits that look like each letter of the alphabet,
and each fold into $\C$.
See \cref{fig:font}, and \url{http://erikdemaine.org/fonts/cubefolding/} for a web app.

\begin{figure}[htbp]
  \centering
  \def\scale{0.22}
  \parskip=3pt
  
  \includegraphics[scale=\scale]{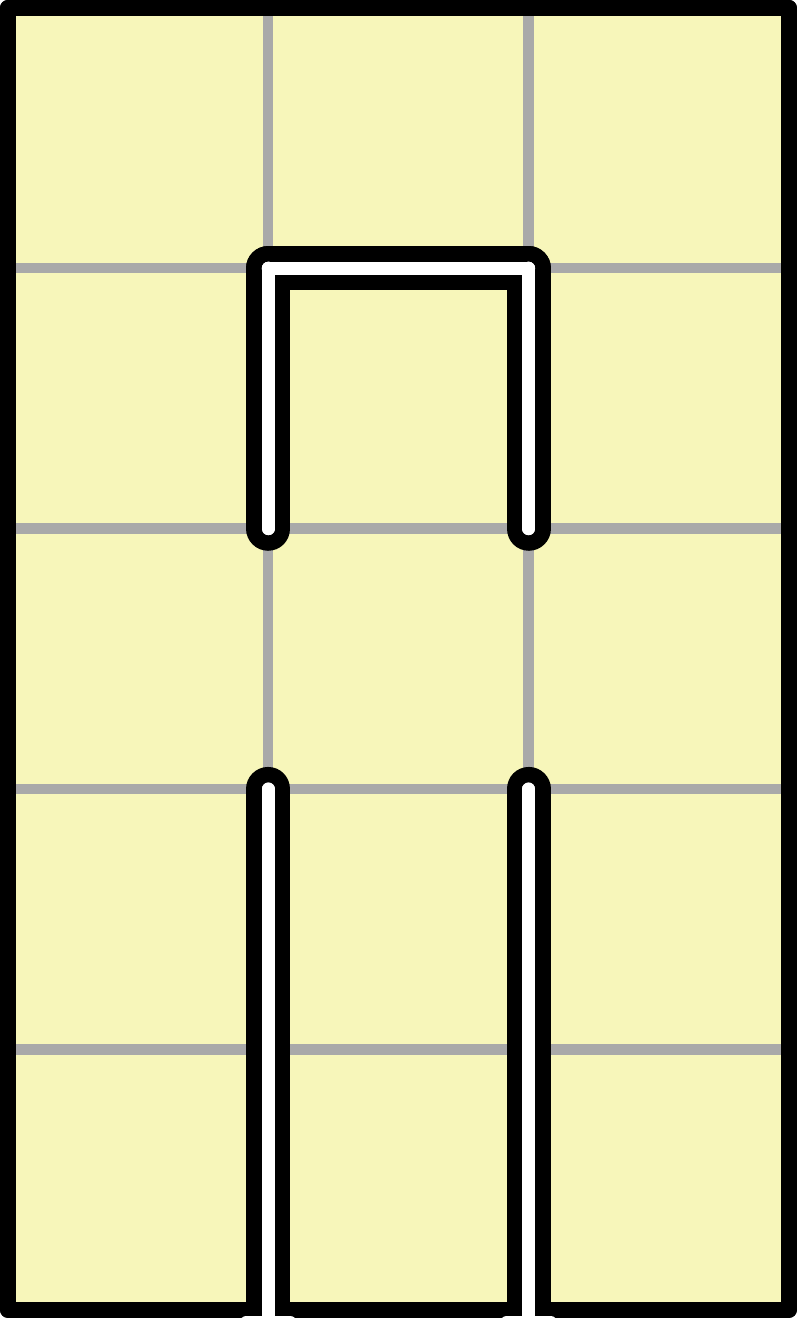}
  \includegraphics[scale=\scale]{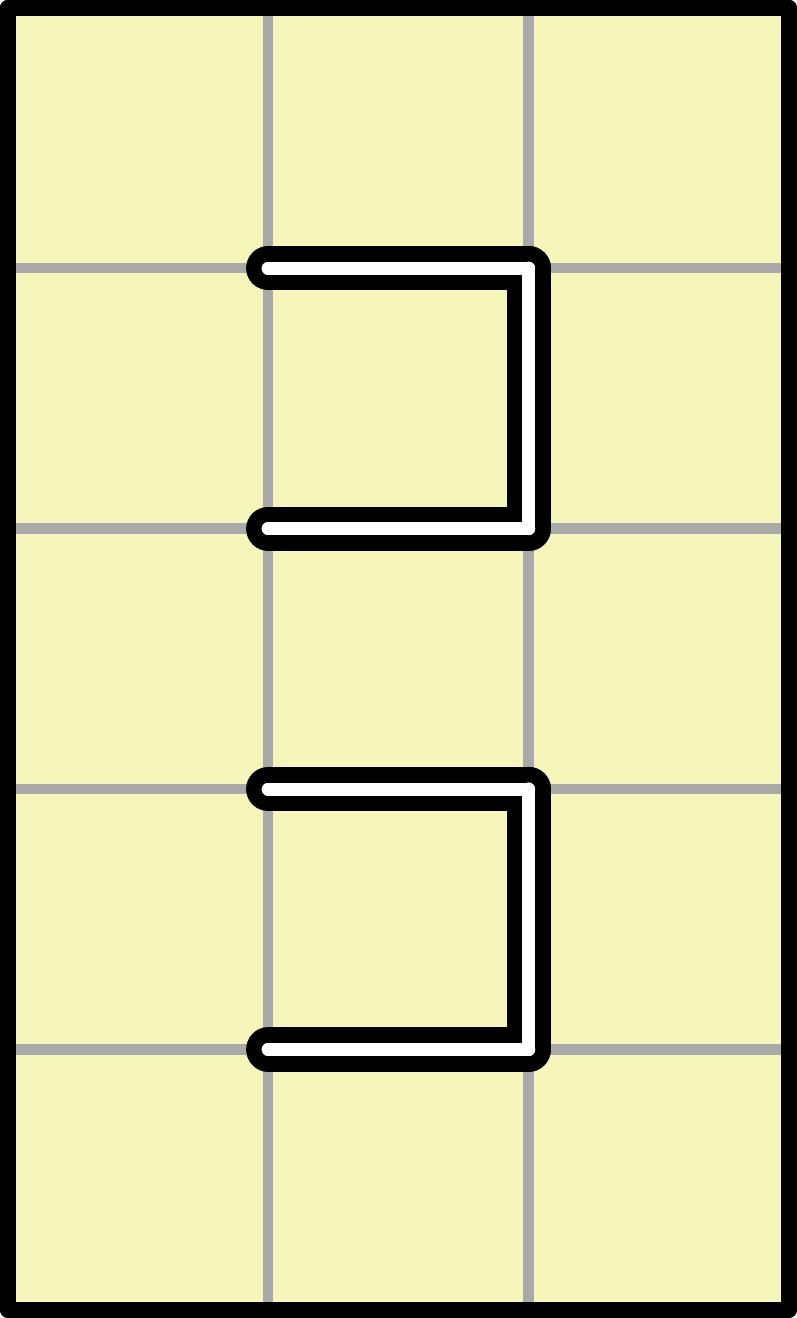}
  \includegraphics[scale=\scale]{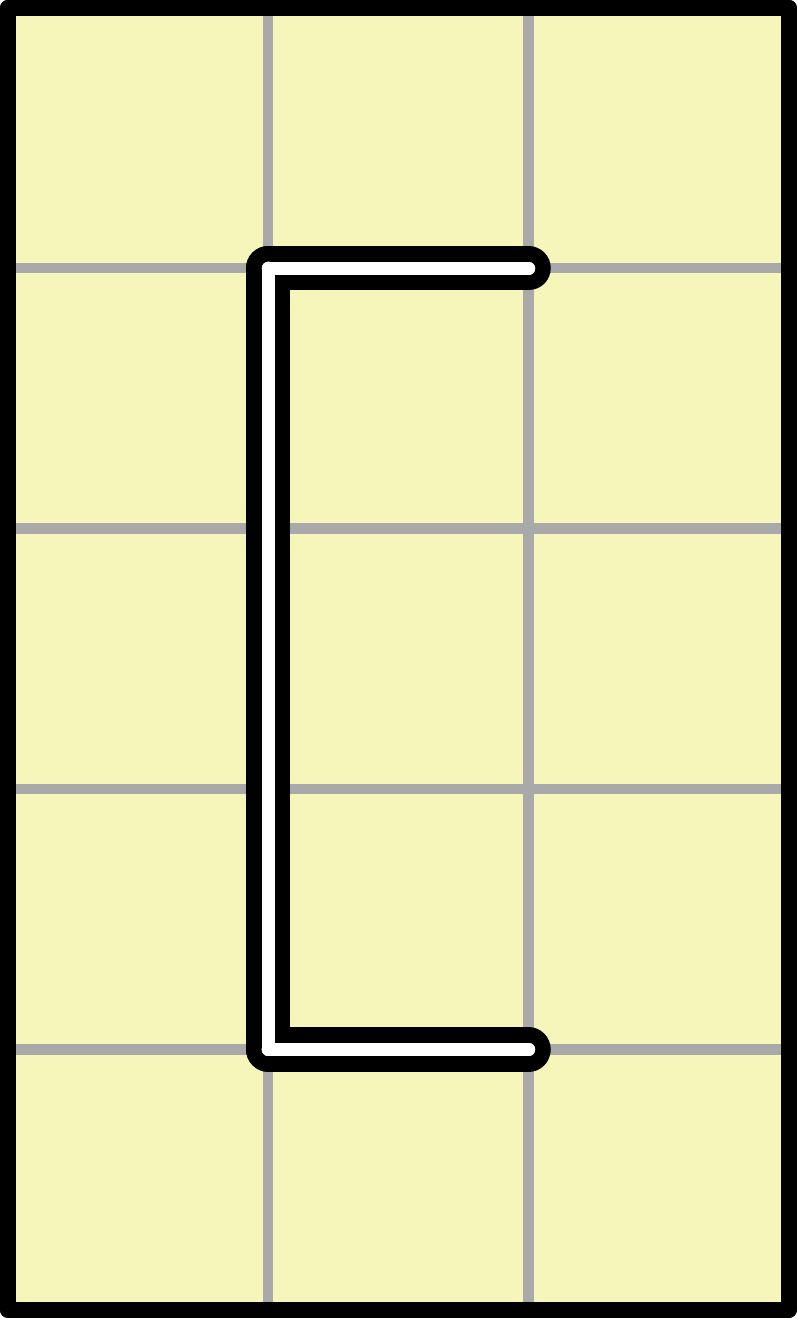}
  \includegraphics[scale=\scale]{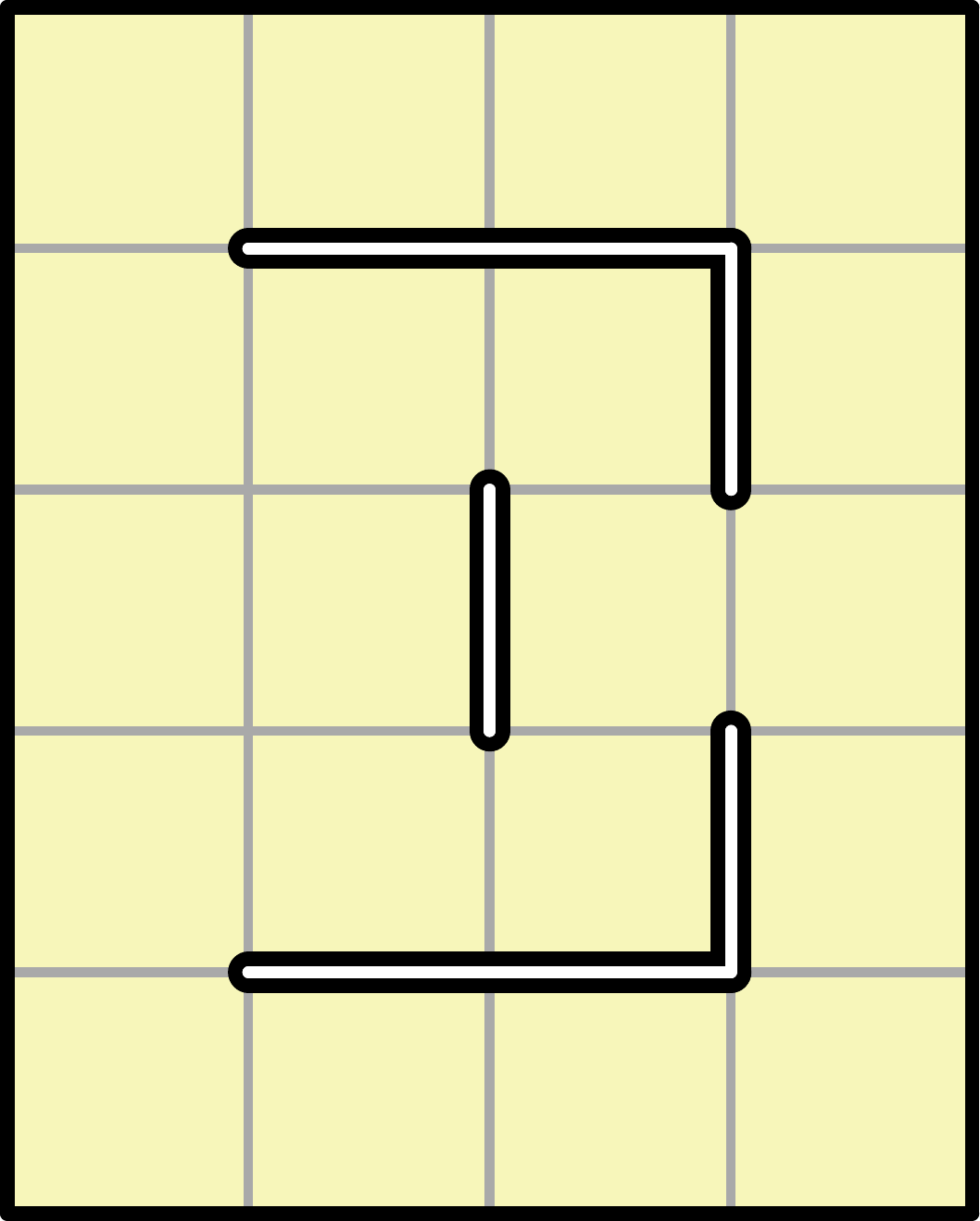}
  \includegraphics[scale=\scale]{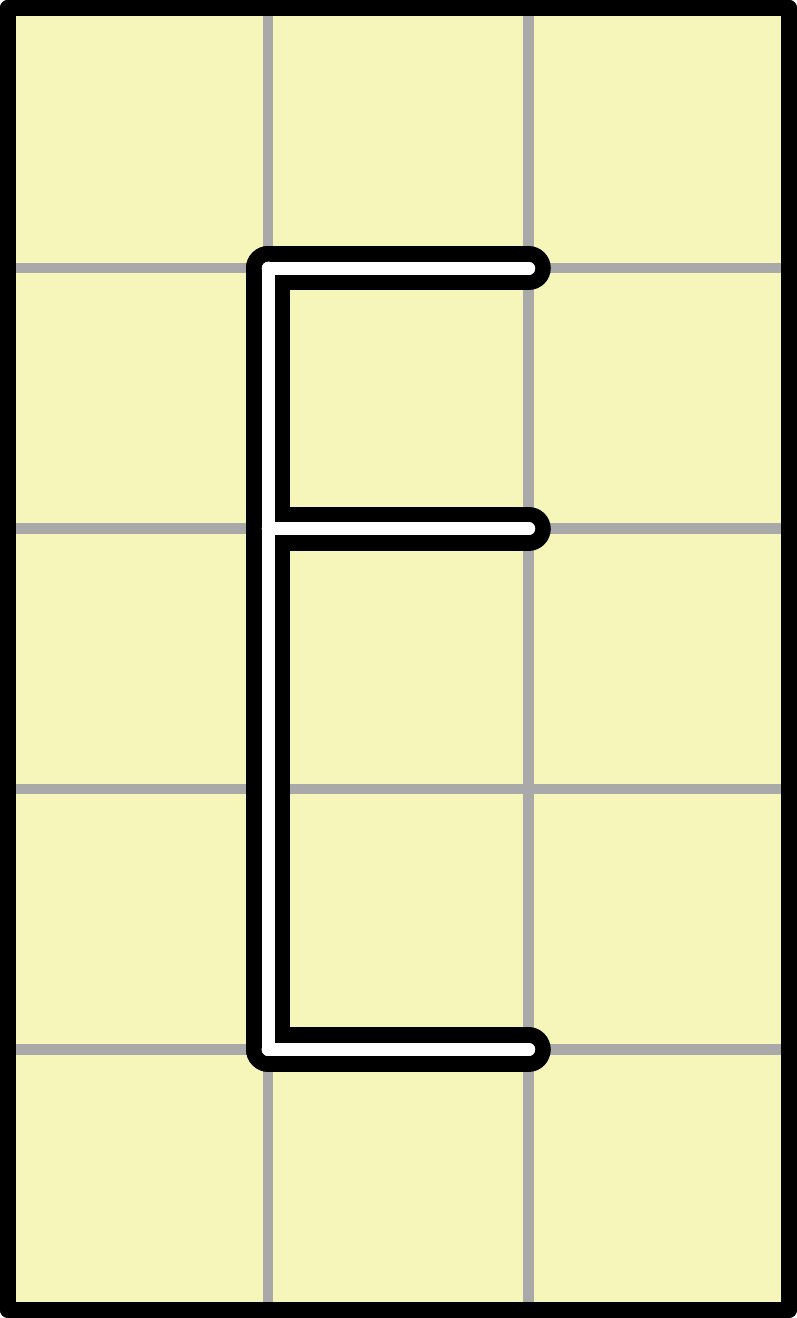}
  \includegraphics[scale=\scale]{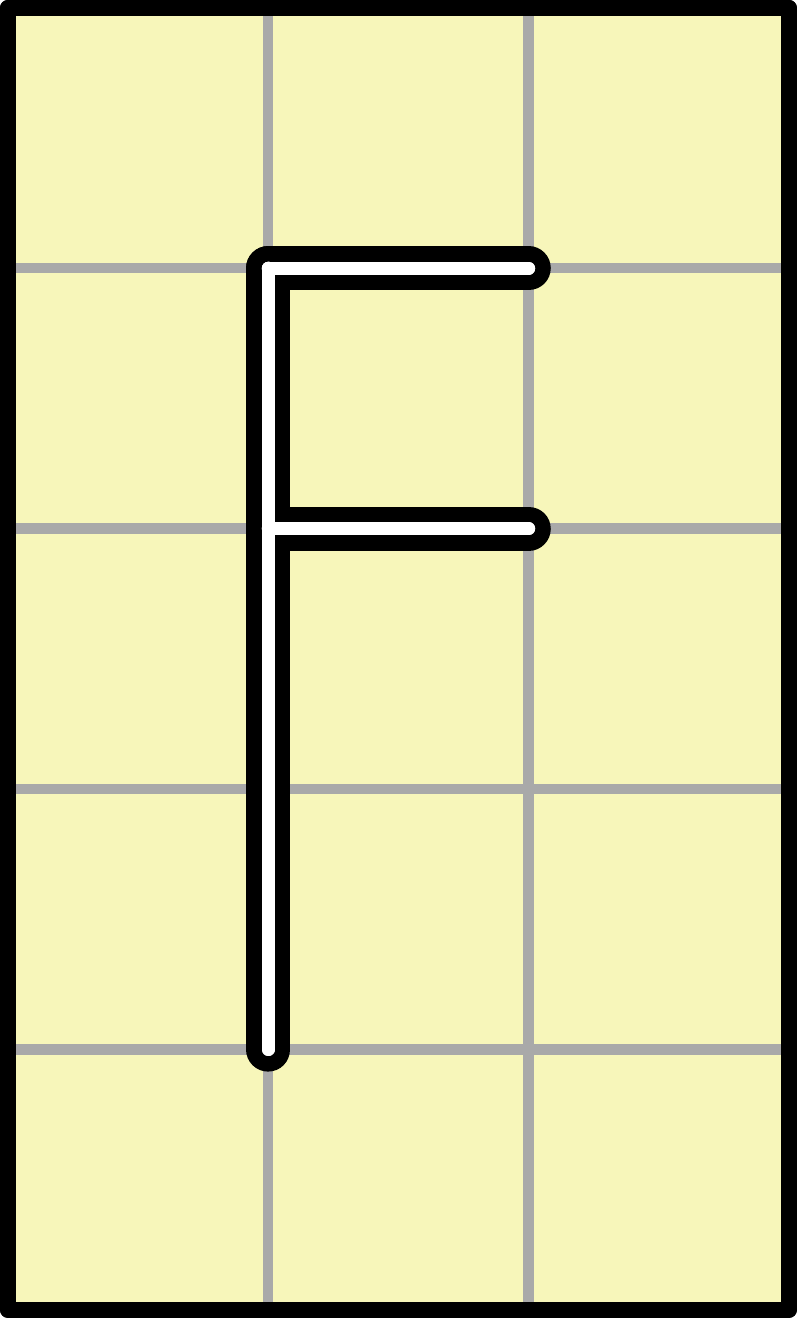}

  \includegraphics[scale=\scale]{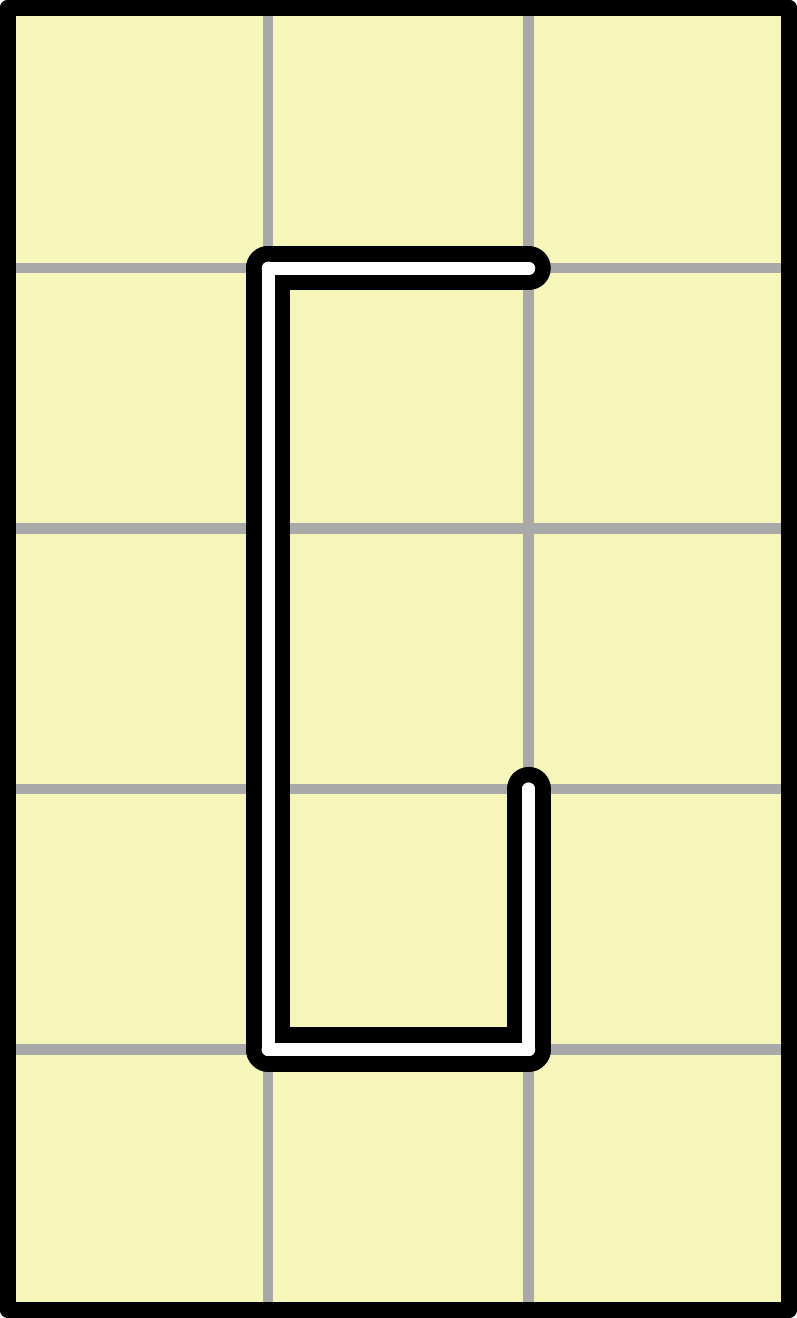}
  \includegraphics[scale=\scale]{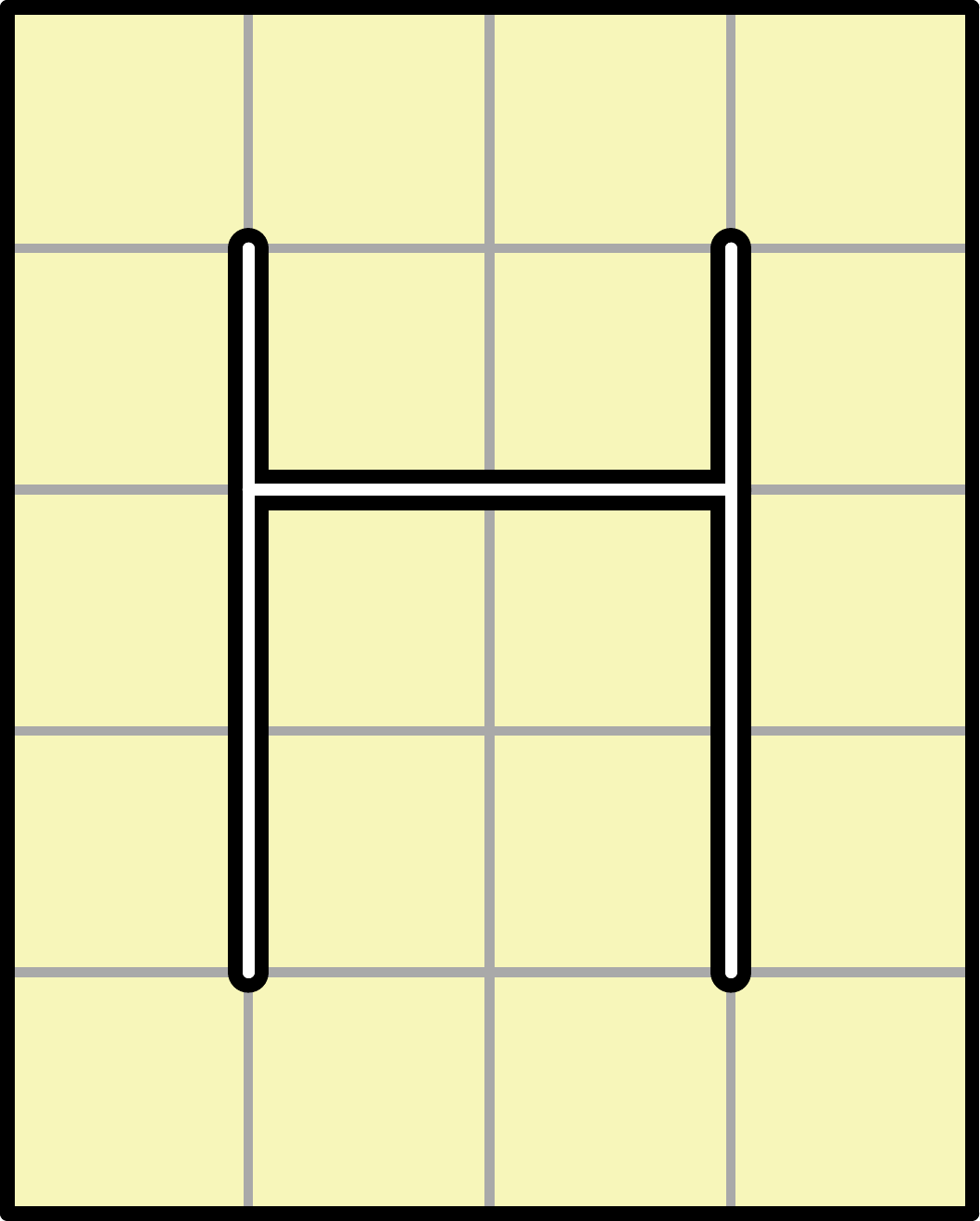}
  \includegraphics[scale=\scale]{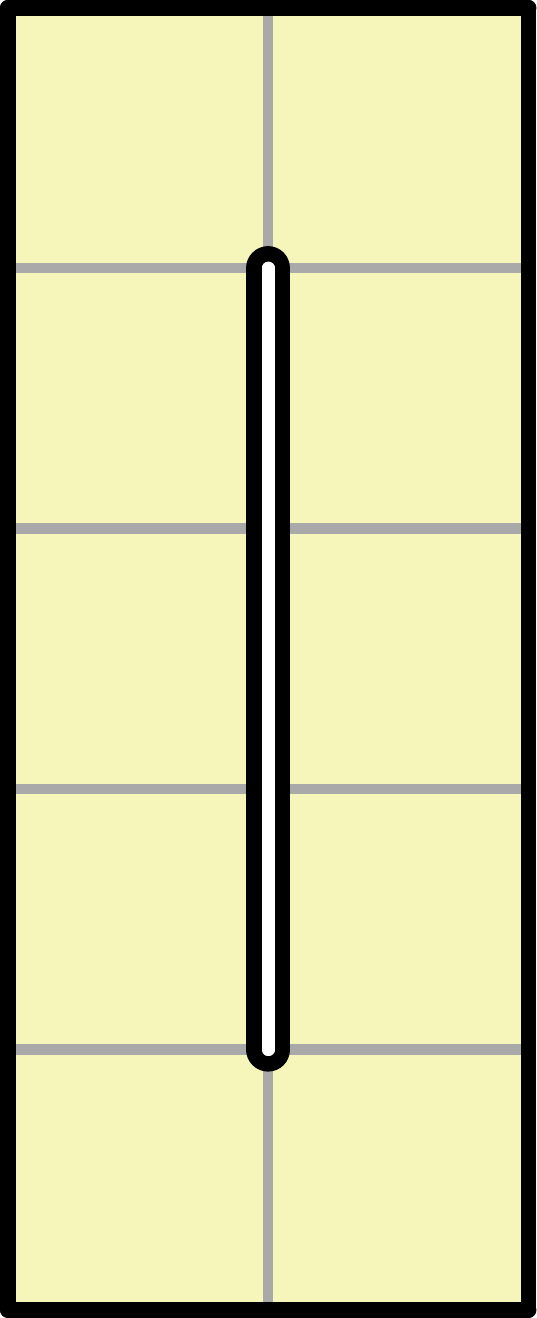}
  \includegraphics[scale=\scale]{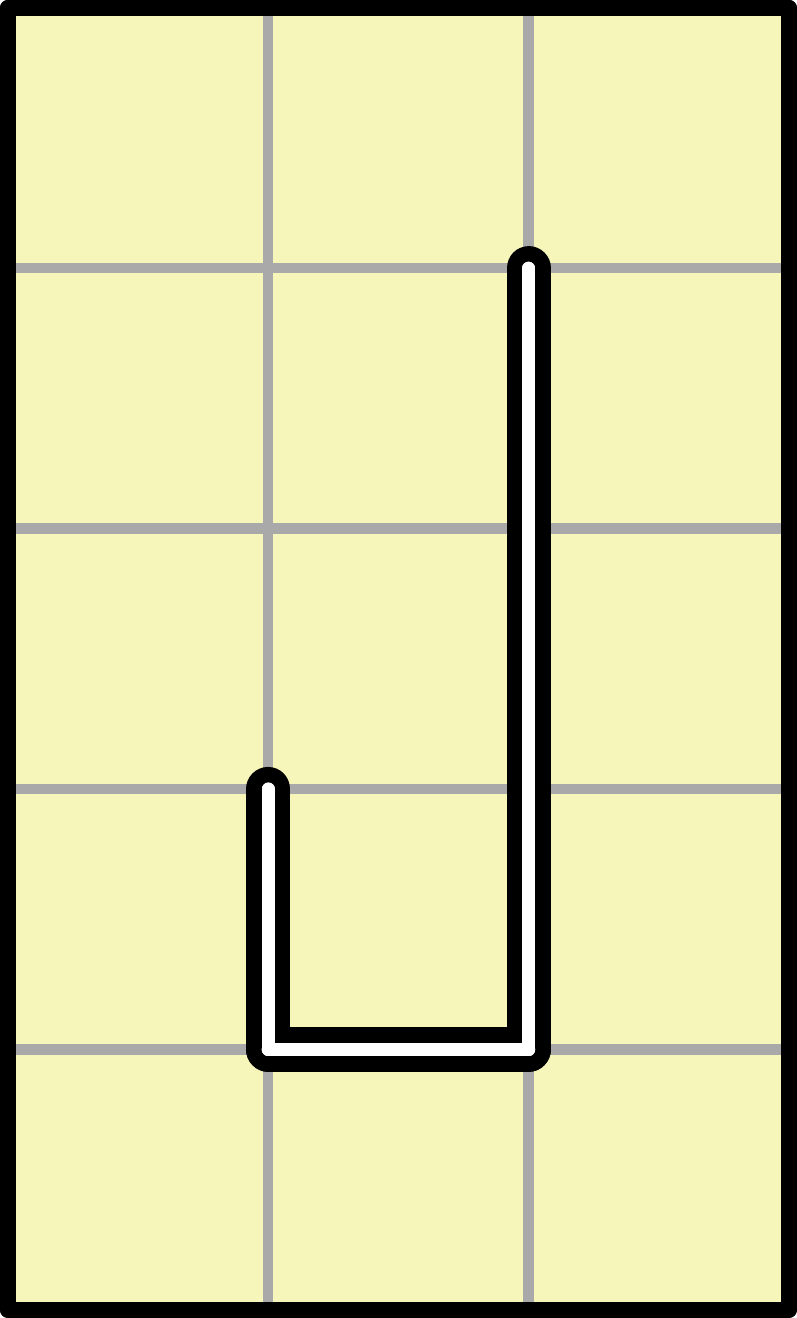}
  \includegraphics[scale=\scale]{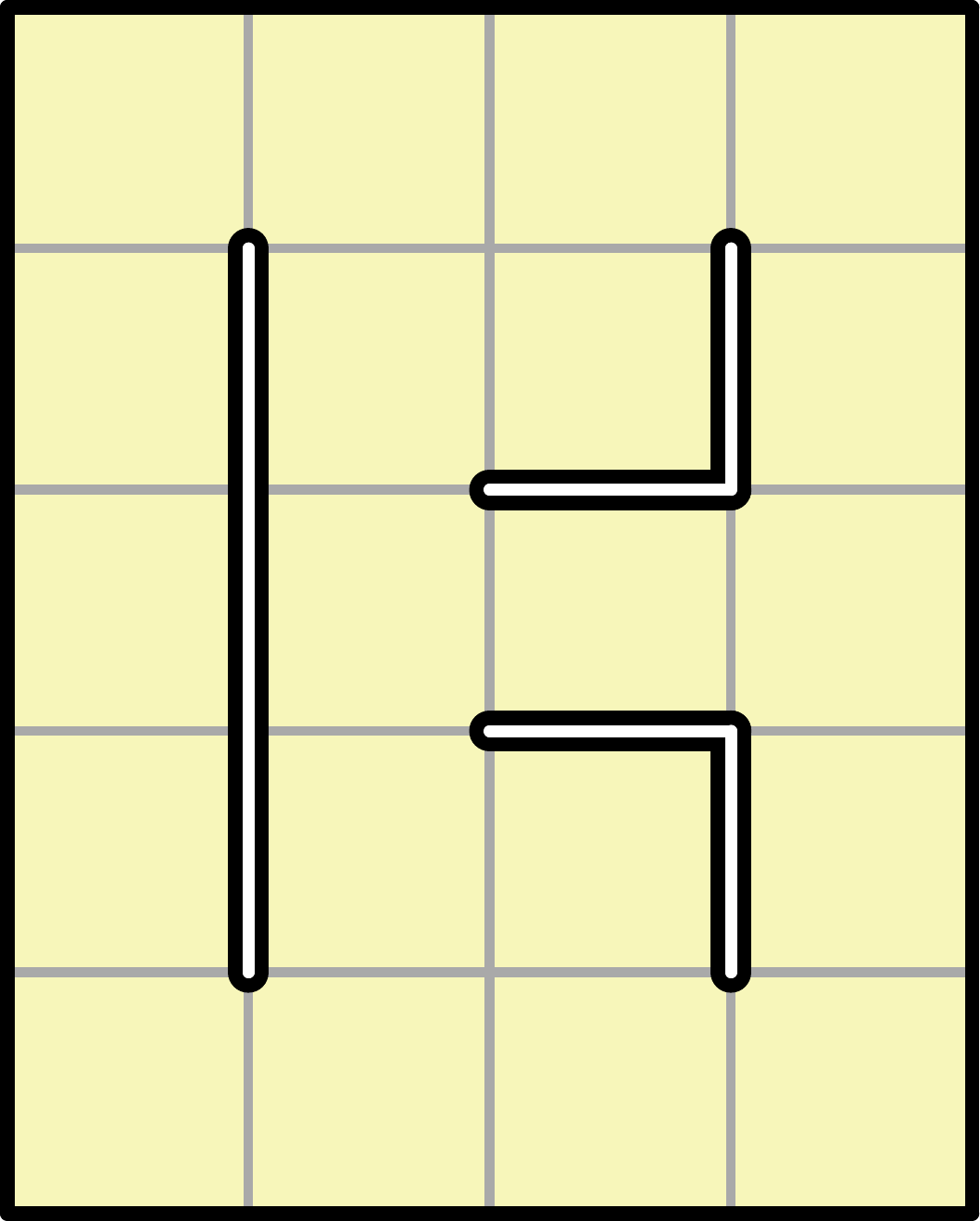}
  \includegraphics[scale=\scale]{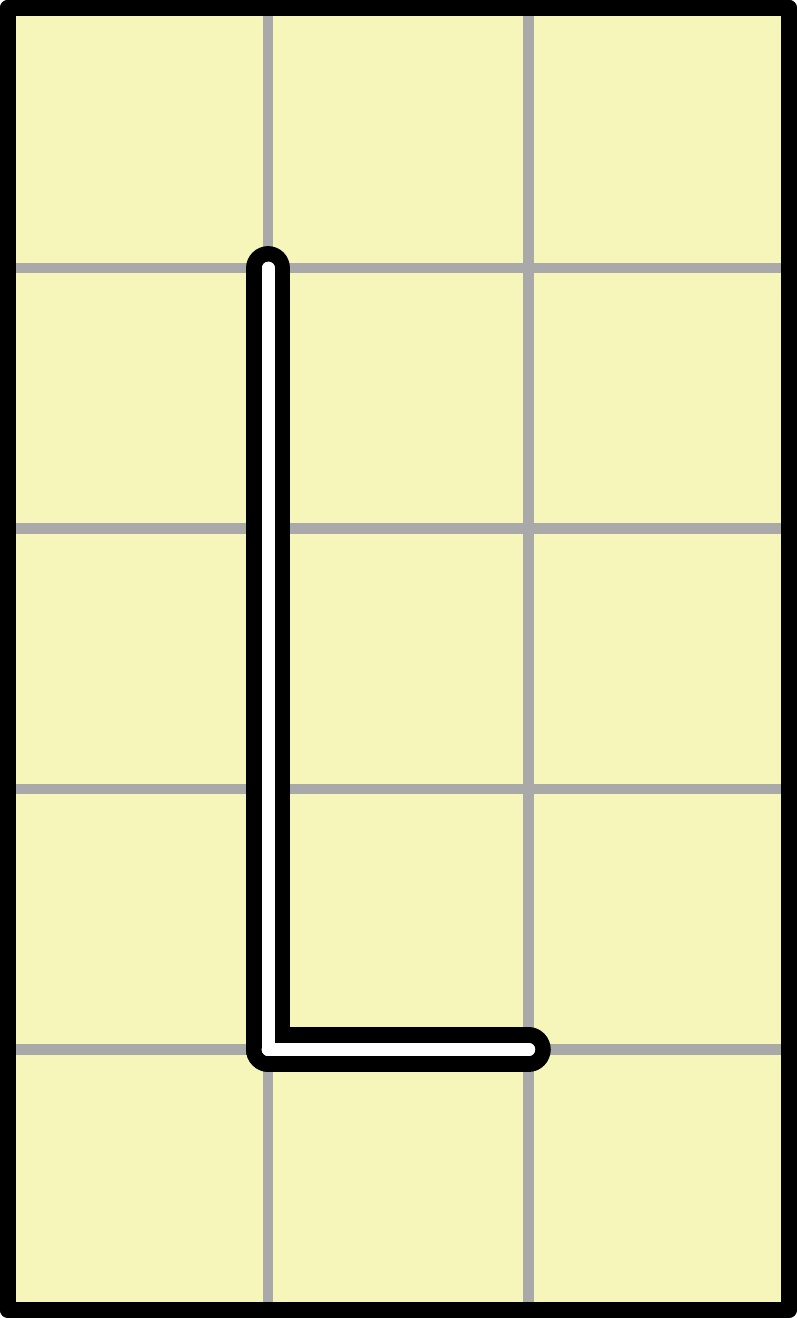}

  \includegraphics[scale=\scale]{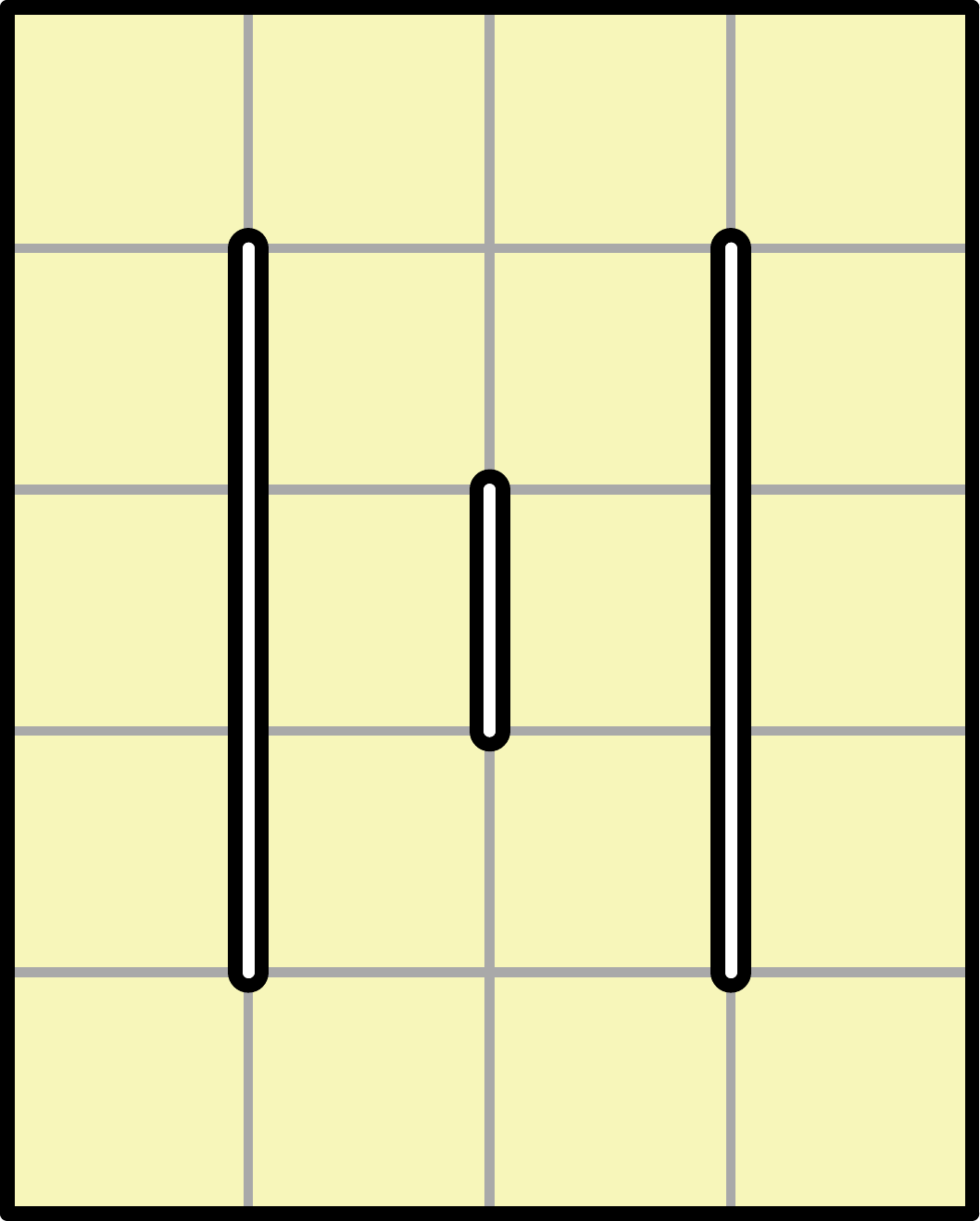}
  \includegraphics[scale=\scale]{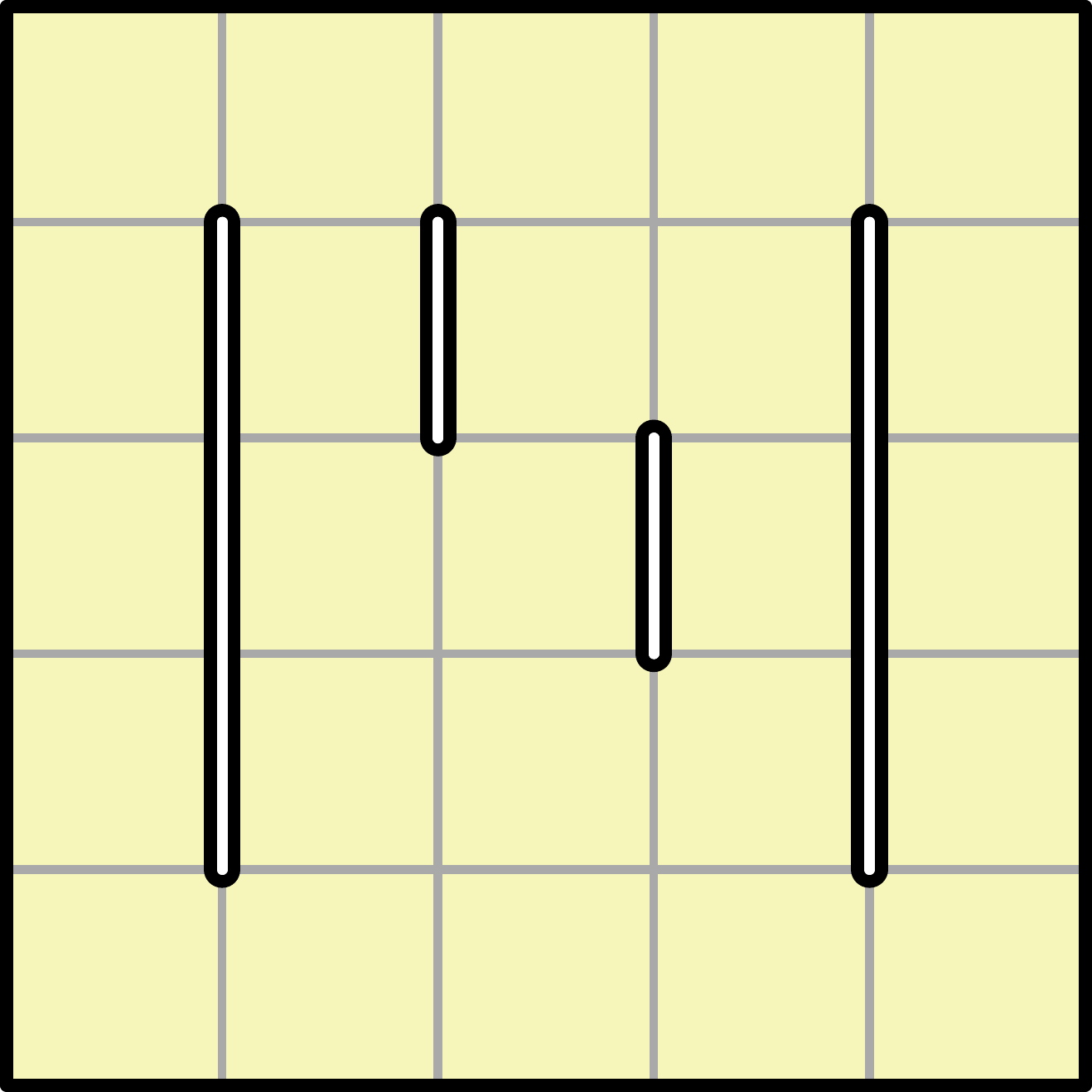}
  \includegraphics[scale=\scale]{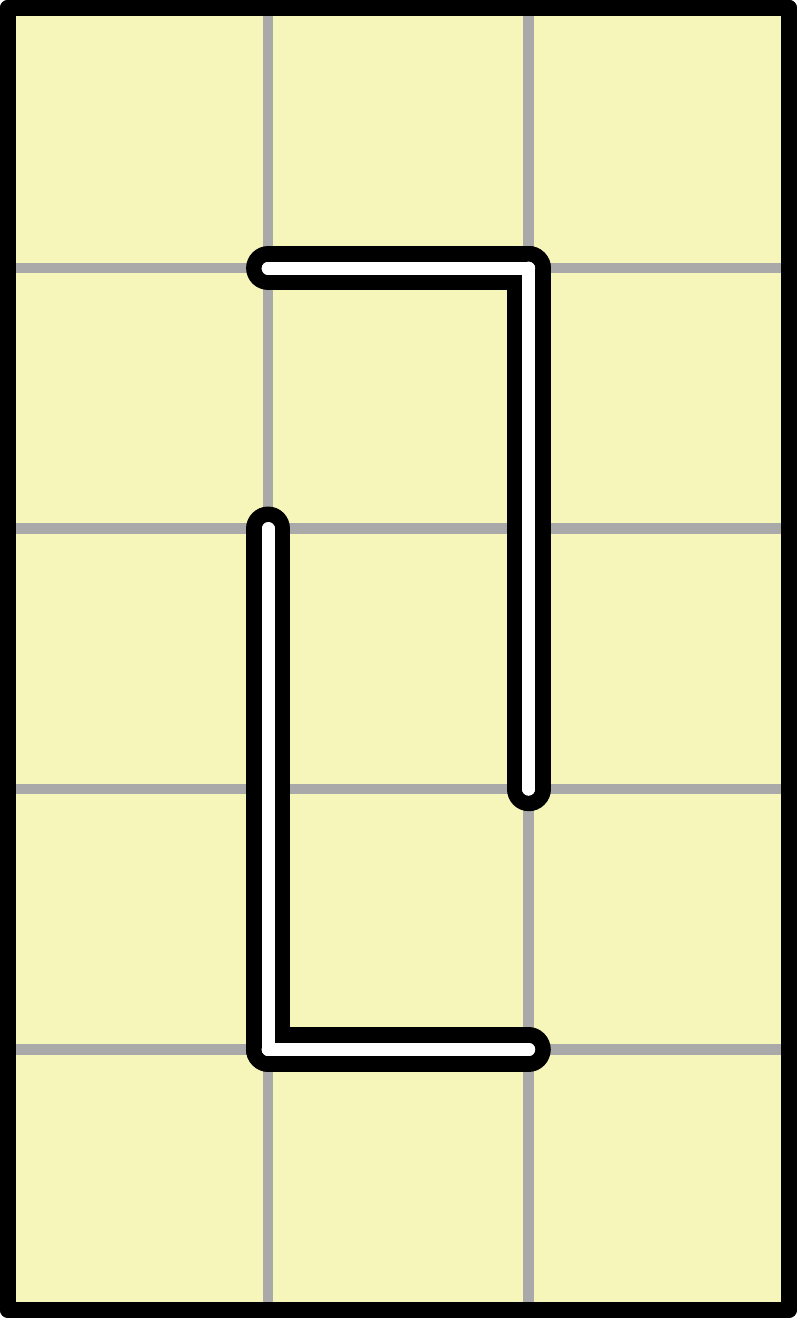}
  \includegraphics[scale=\scale]{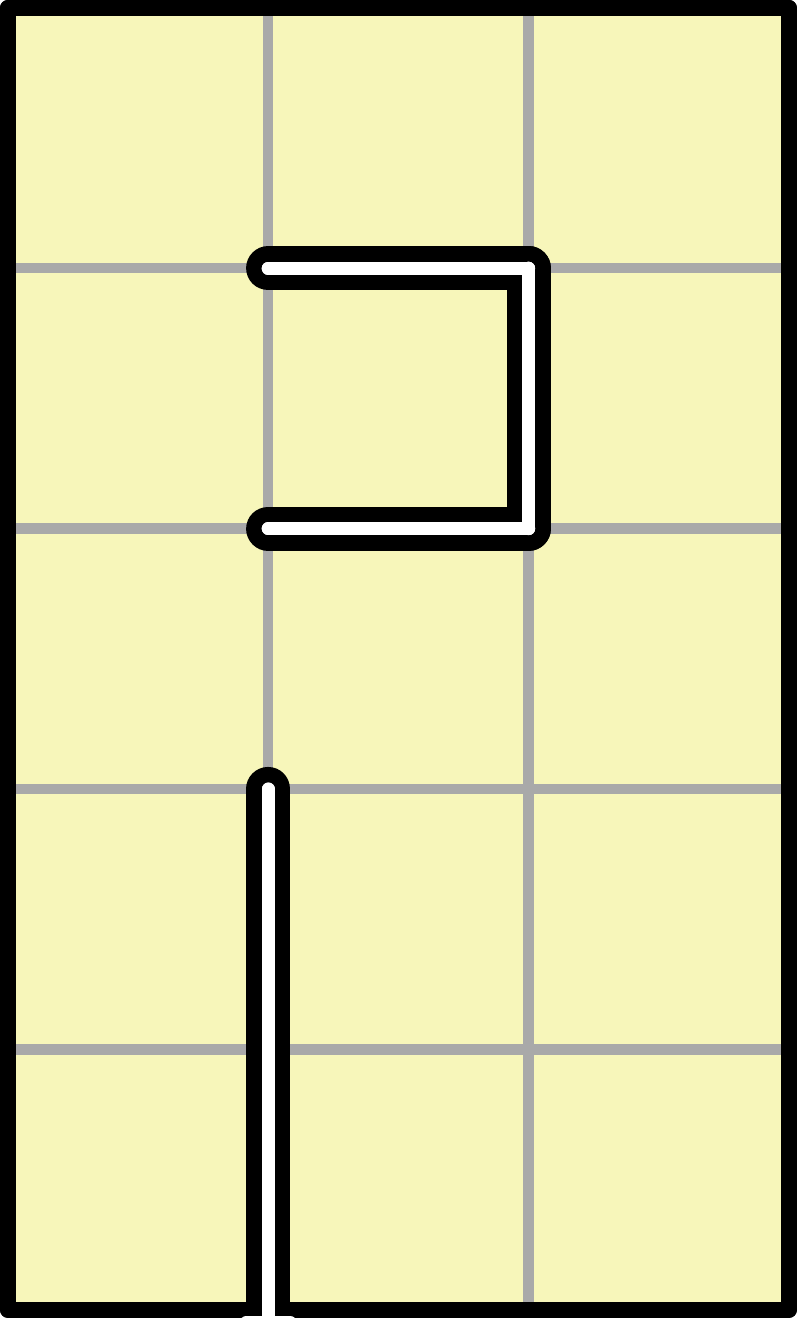}
  \includegraphics[scale=\scale]{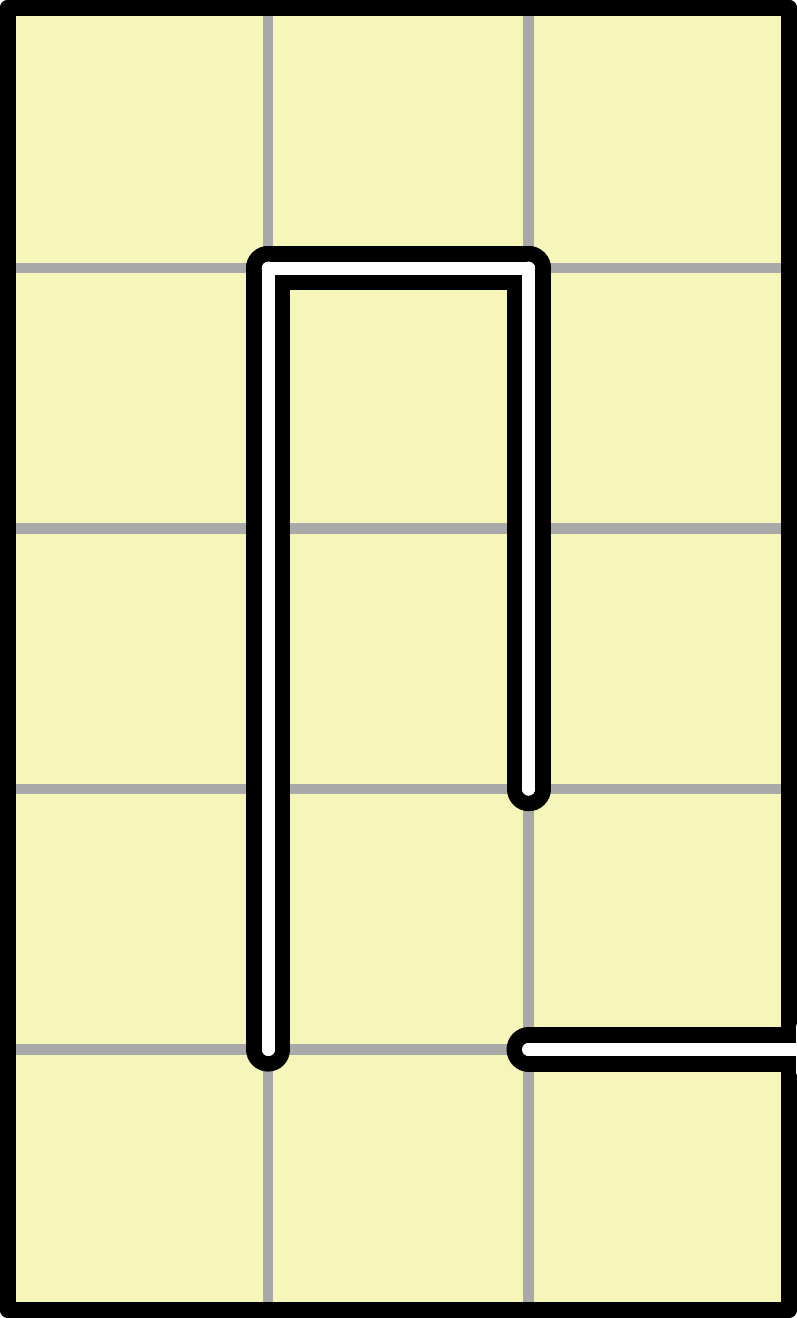}
  
  \includegraphics[scale=\scale]{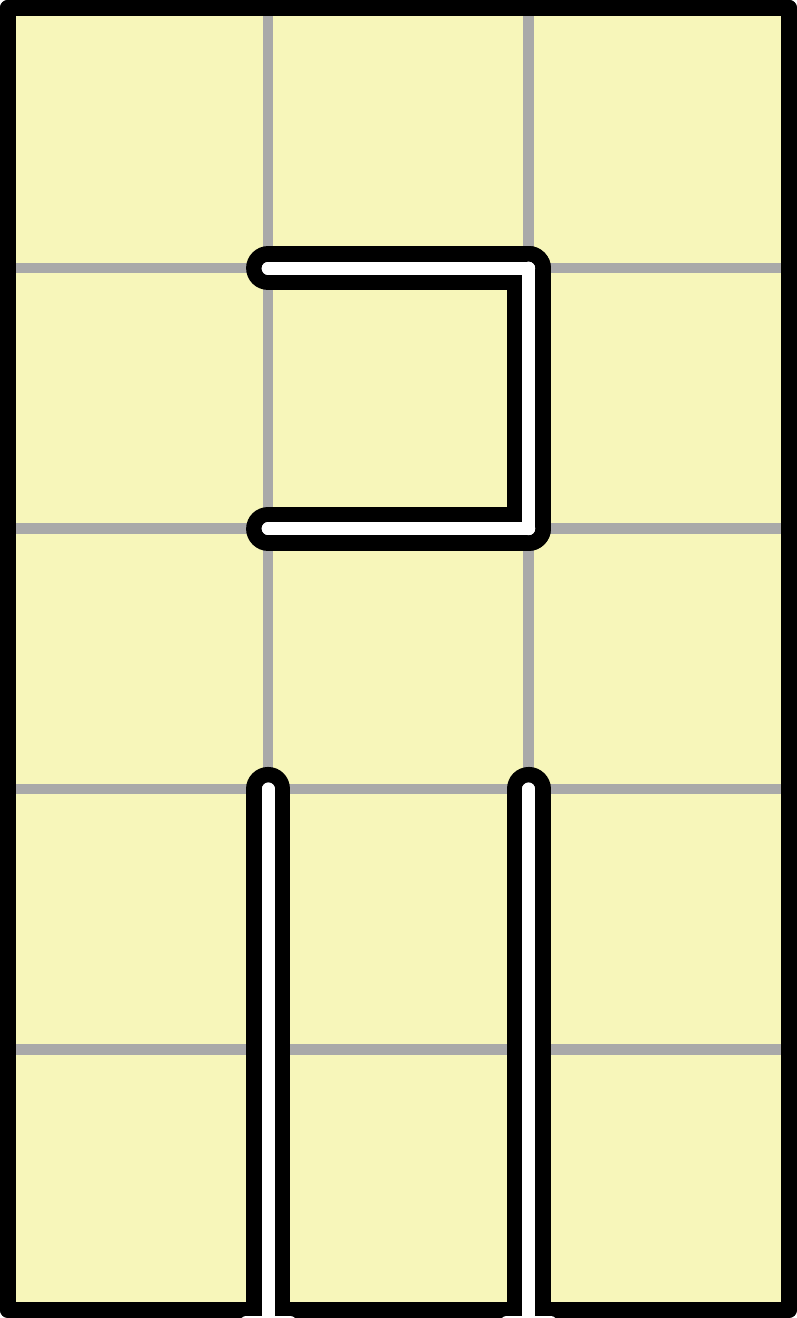}
  \includegraphics[scale=\scale]{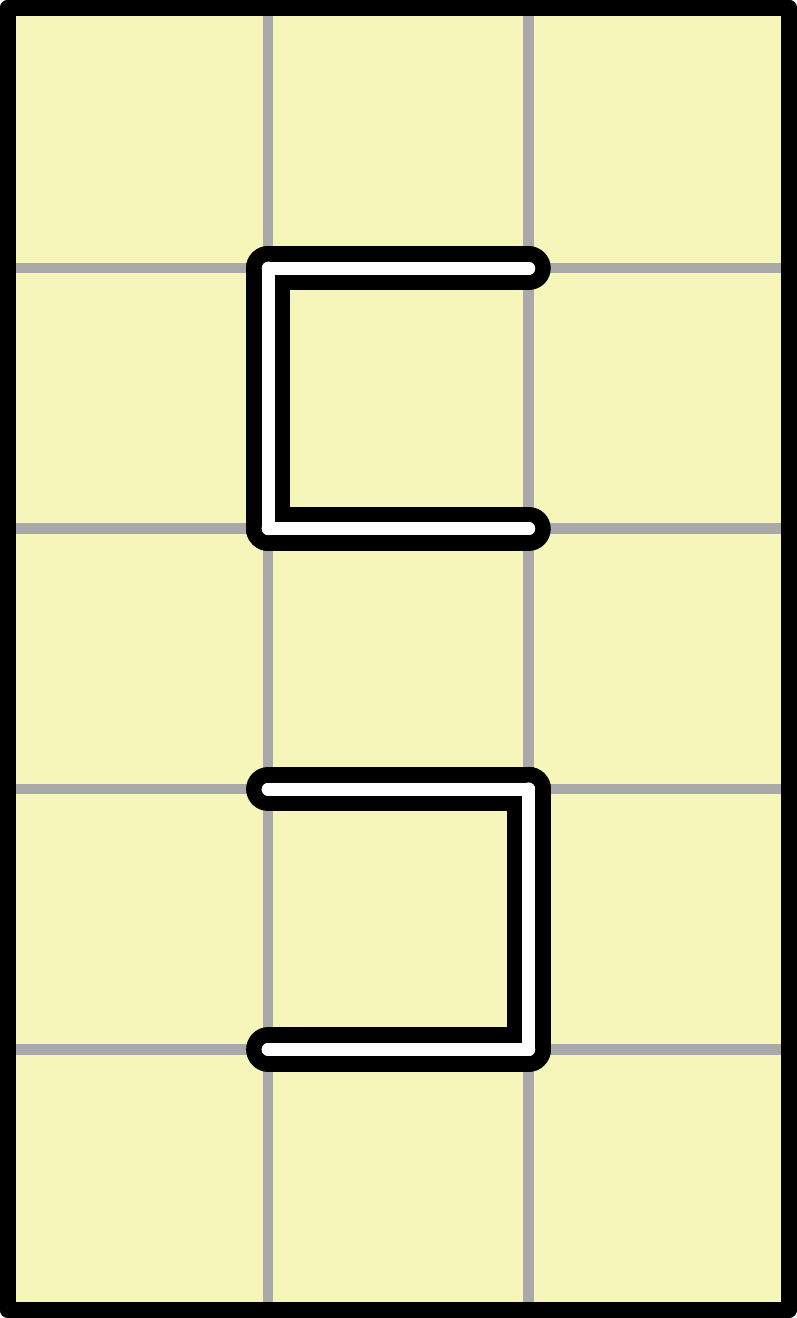}
  \includegraphics[scale=\scale]{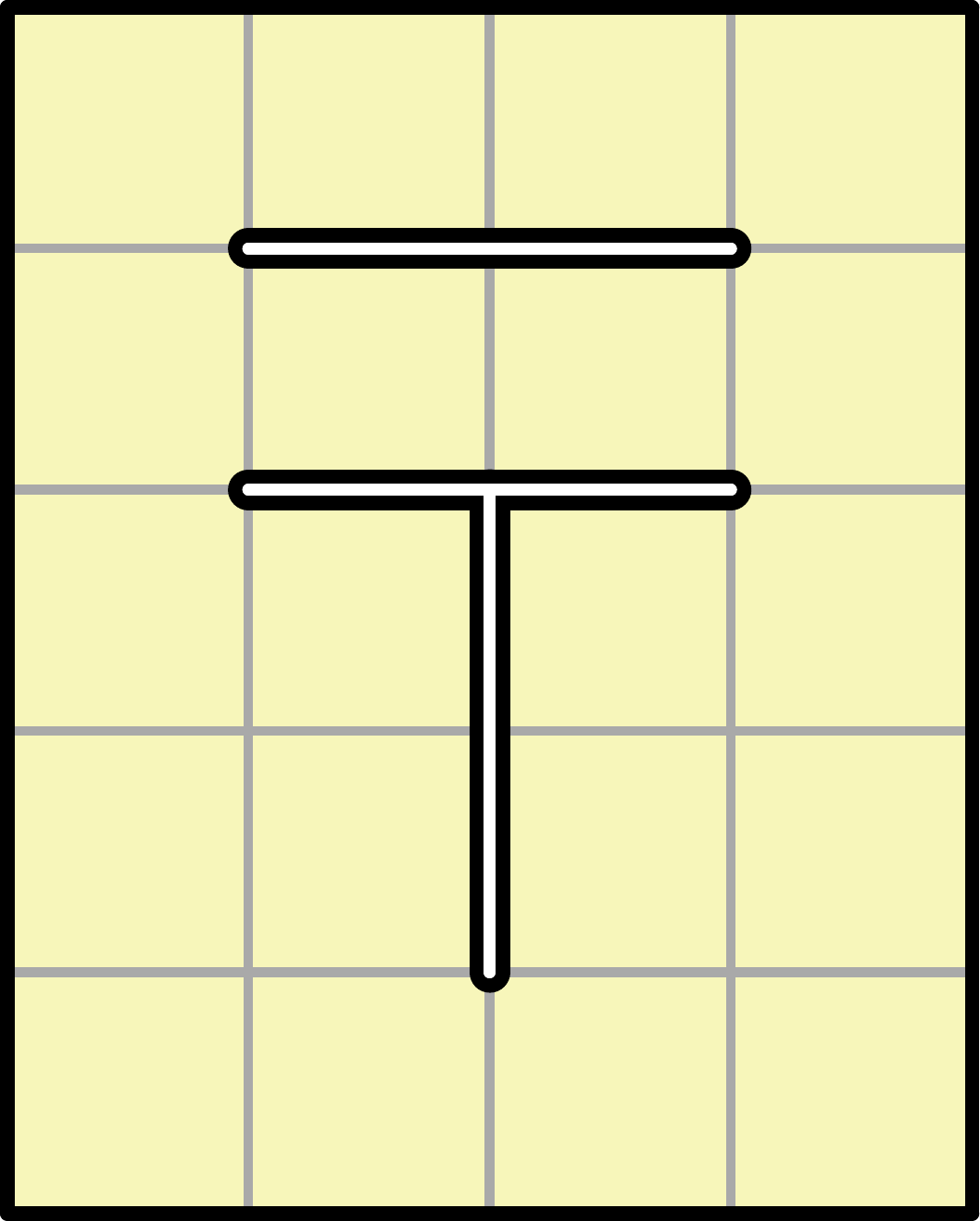}
  \includegraphics[scale=\scale]{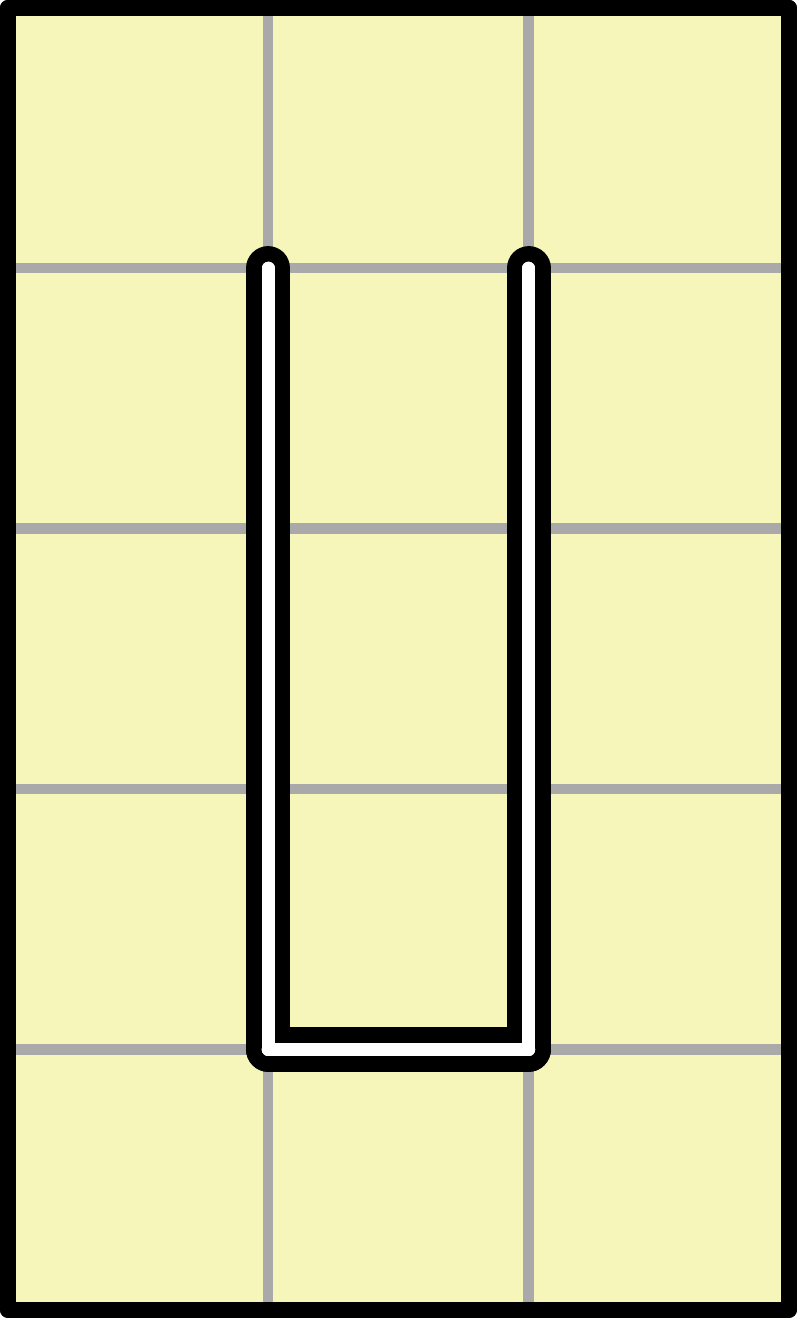}
  \includegraphics[scale=\scale]{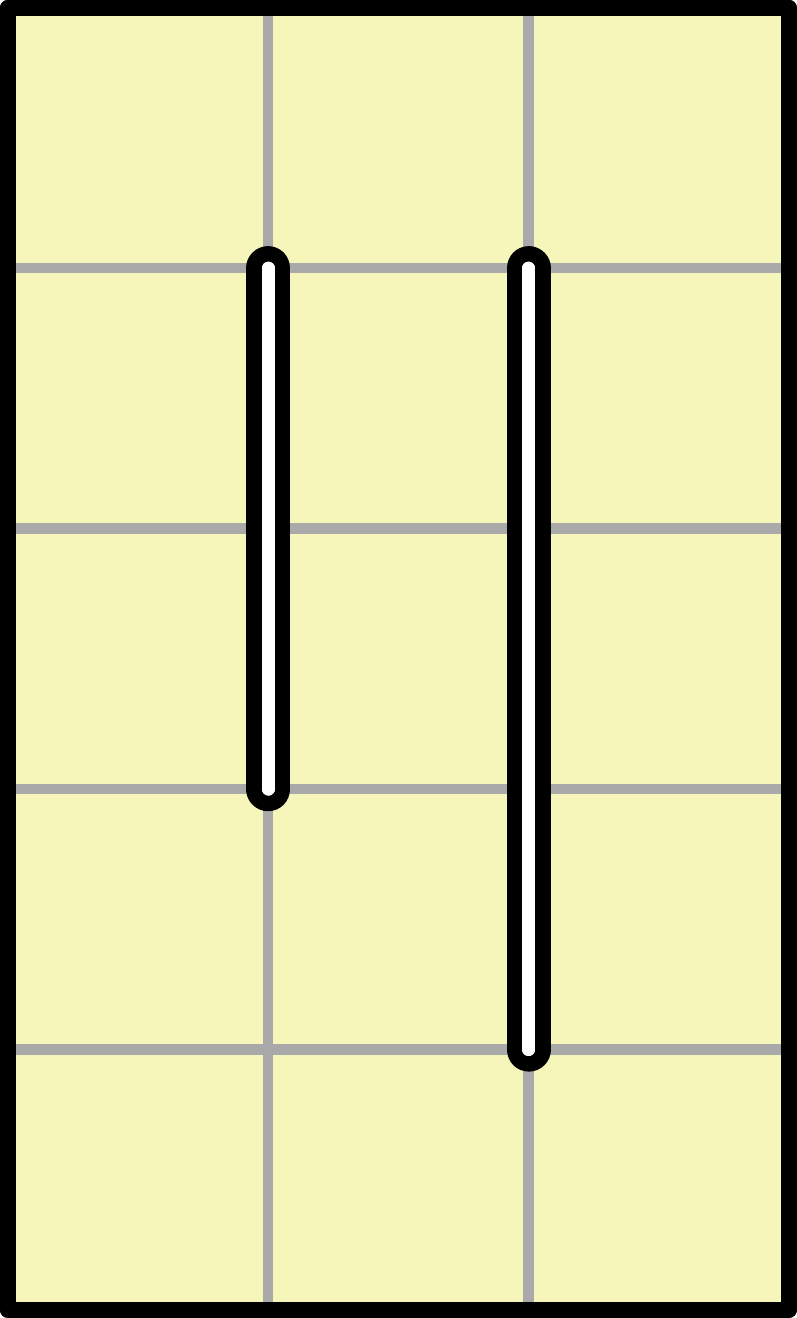}

  \includegraphics[scale=\scale]{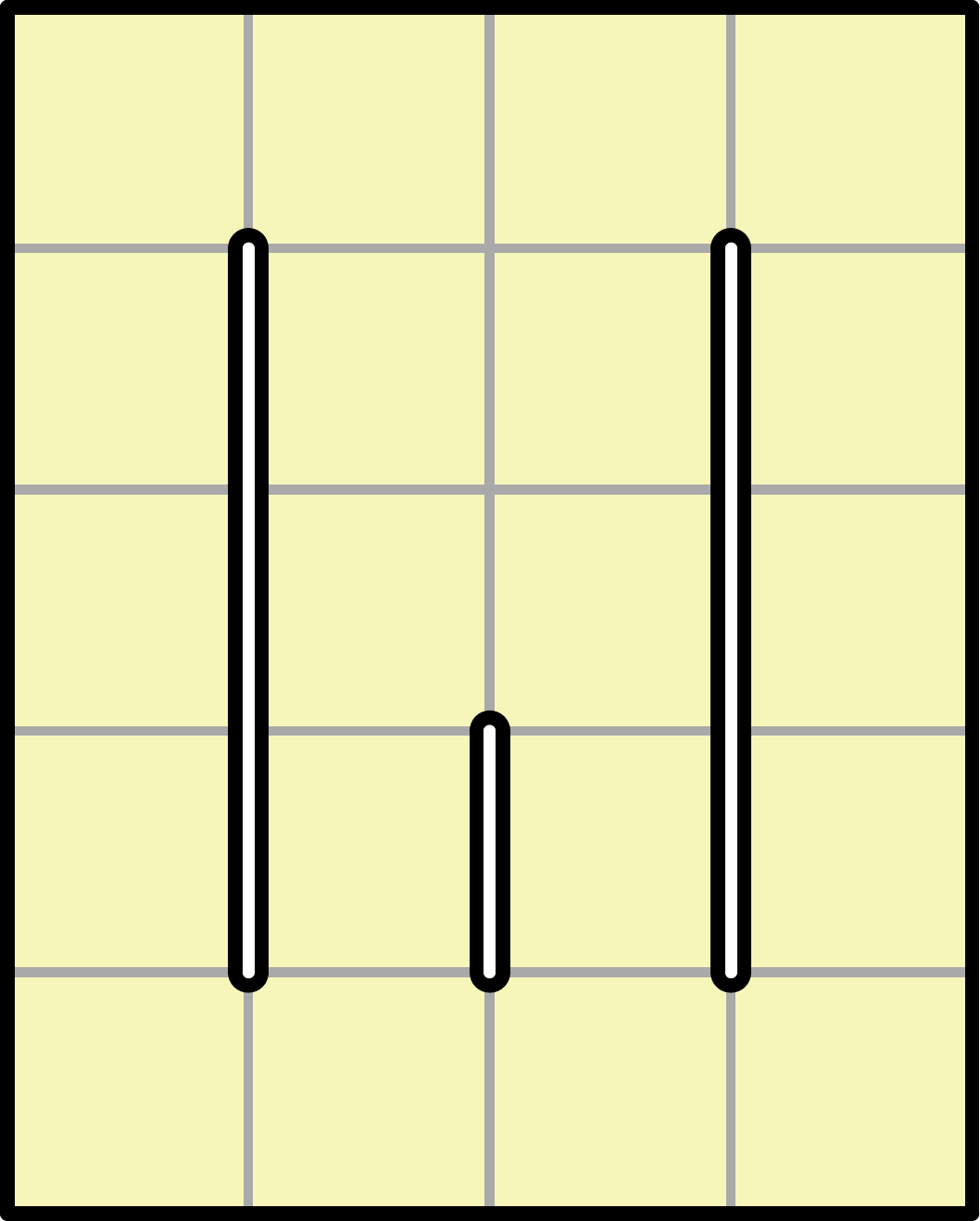}
  \includegraphics[scale=\scale]{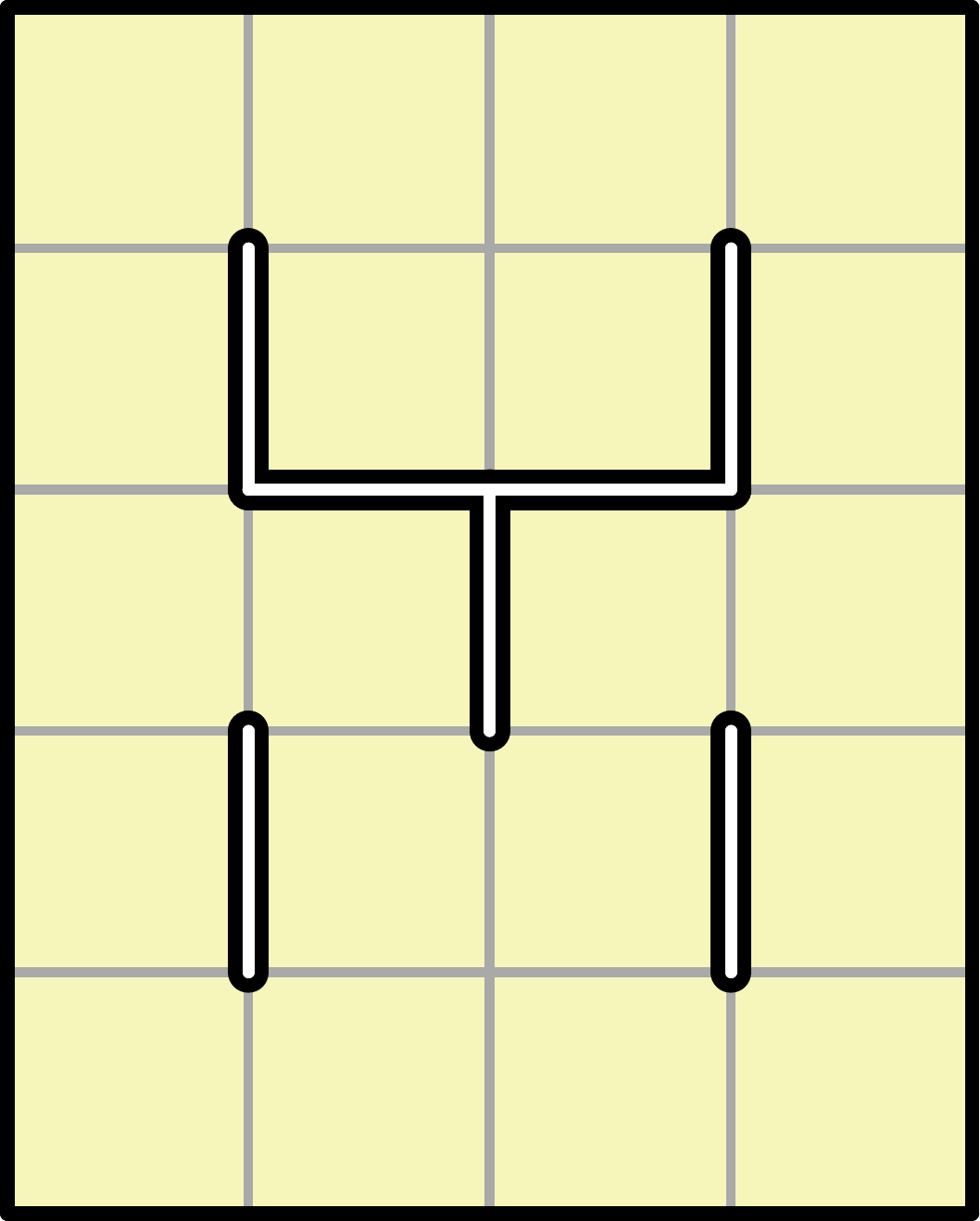}
  \includegraphics[scale=\scale]{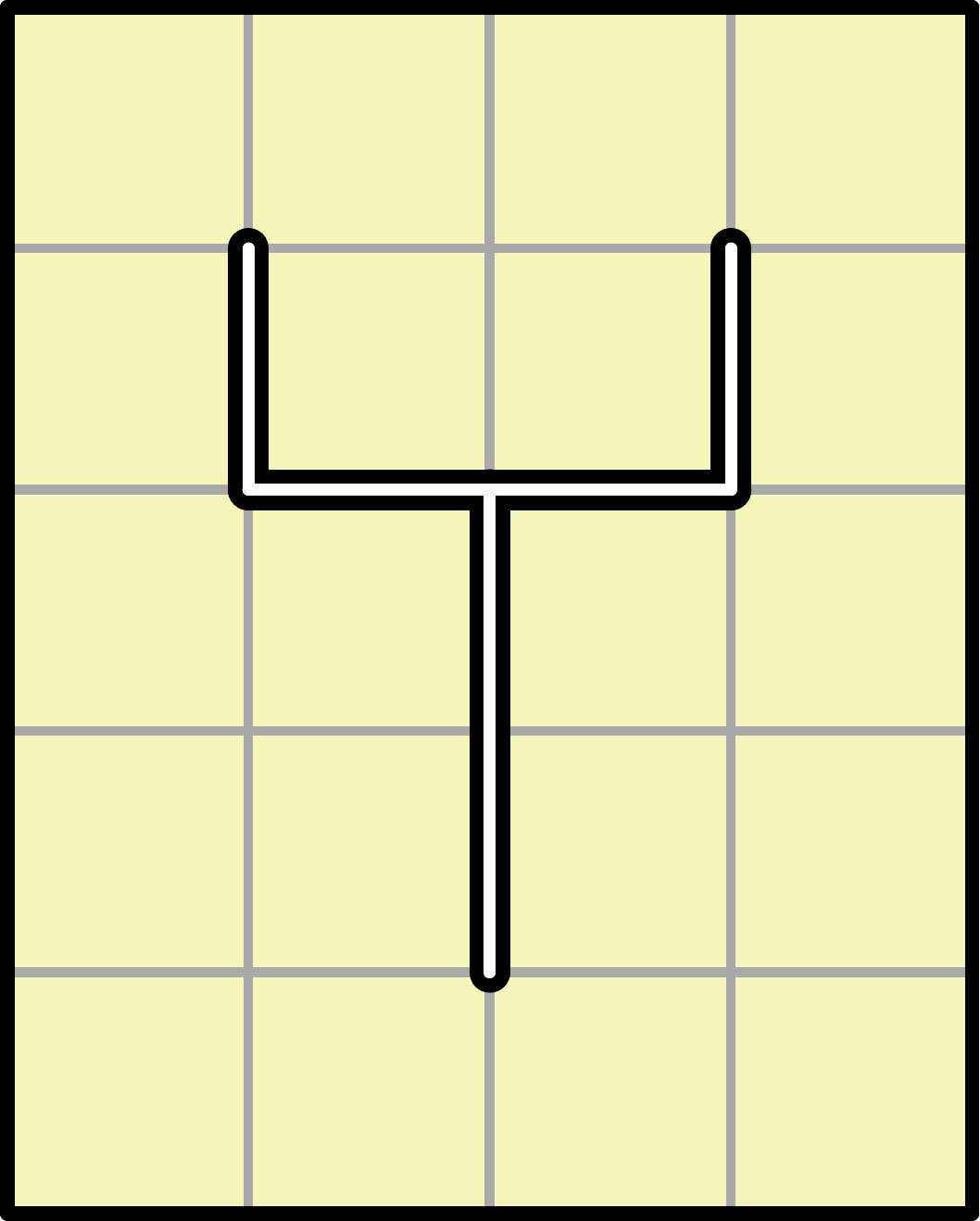}
  \includegraphics[scale=\scale]{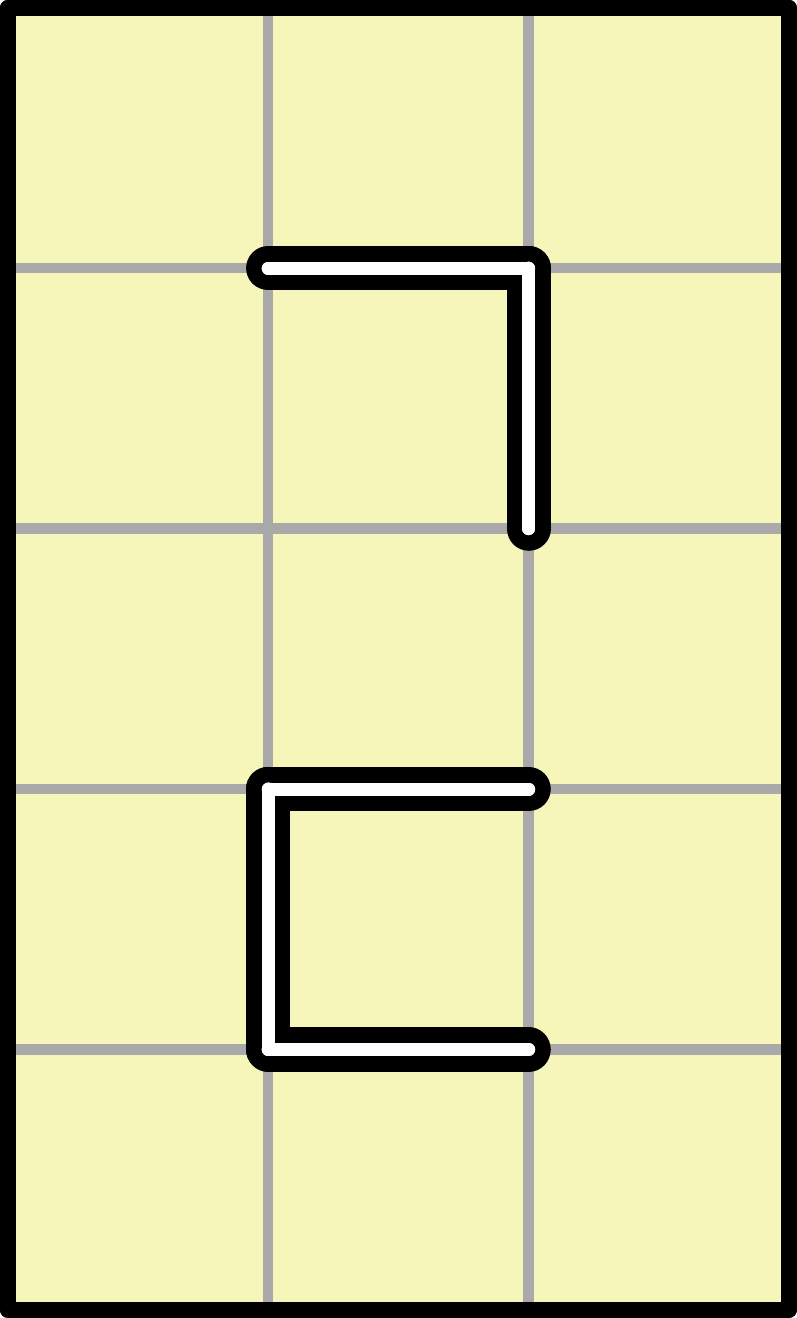}
  
  \caption{Cube-folding font: the slits representing each letter enable each rectangular puzzle to fold into a cube.}
  \label{fig:font}
\end{figure}

\medskip
We conclude with a list of interesting open problems:
\begin{compactitem}
\item Does a consistent grid-point mapping output by the algorithm in \cref{subsec:alg} imply that the polyomino is foldable? If so, is the folding uniquely determined?
\item Is any rectangular polyomino with one L-slit, U-slit, or straight slit of size~2  foldable? Currently, we only know that the small polyominoes in  \cref{fig:simple_slits_cases} do not fold.
\item We considered the existence of only a folded state in the shape of~$\C$,
  but what if we require a continuous folding motion from the unfolded
  polyomino into~$\C$?  These two models are known to be equivalent for
  polygons without holes \cite{PaperReachability_CCCG2001,GFALOP}, but
  equivalence remains an open problem for polygons with holes as in our case.
\end{compactitem}

\section*{Acknowledgments}
This research was performed in part at the 33rd Bellairs Winter Workshop on Computational Geometry. We thank all other participants for a fruitful atmosphere.
H. Akitaya was supported by NSF CCF-1422311 \& 1423615. Z. Mas\'arov\'a was partially funded by Wittgenstein Prize, Austrian
Science Fund (FWF), grant no. Z 342-N31.

\newpage
\bibliographystyle{elsarticle-num}
\bibliography{folding}

\end{document}